\documentclass{easychair}

\usepackage{amsmath}
\usepackage{amsfonts}
\usepackage{pstricks}
\usepackage{pst-node}
\usepackage[draft]{minted}

\newif\ifdraft\draftfalse

\newcommand\easyicalp[2]{#1}

\usepackage{amsfonts}
\usepackage{amssymb}
\usepackage{bbm}
\usepackage{subfig}
\usepackage{wrapfig}
\usepackage{xspace}

\newcommand\etal{~\textit{et al.}\xspace}
\newcommand\mnew{\text{new}}
\newcommand\cshore{C-SHORe\xspace}
\newcommand\trecs{TRecS\xspace}
\newcommand\gtrecs{GTRecS\xspace}

\newcommand\travmc{TravMC\xspace}

\newcommand\horsat{HorSat\xspace}
\newcommand\preface{Preface\xspace}

\newcommand\idxi{i}
\newcommand\idxj{j}
\newcommand\numof{m}
\newcommand\numofl{l}

\newcommand\lang{\mathcal{L}}

\newcommand\langof[1]{\ap{\lang}{#1}}

\newcommand\setcomp[2]{\left\{{#1}\ \left|\ {#2}\right.\right\}}

\newcommand\set[1]{\left\{{#1}\right\}}

\newcommand\tup[1]{\brac{#1}}

\newcommand\brac[1]{\left({#1}\right)}

\newcommand\card[1]{\left|{#1}\right|}

\newcommand\seqlen[1]{\card{#1}}

\newcommand\ap[2]{{#1}\mathord{\brac{#2}}}

\newcommand\tree{t}
\newcommand\treelabels{\Gamma}
\newcommand\treedom{\mathcal{D}}
\newcommand\treelabelling{\lambda}
\newcommand\tdiri{i}
\newcommand\tdirj{j}
\newcommand\tnode{v}
\newcommand\tnodeset{V}
\newcommand\treelab{\gamma}
\newcommand\maxarity{d}
\newcommand\treeanc{\preceq}
\newcommand\troot{\varepsilon}

\newcommand\treemod[3]{{#1}\mathord{\left[{#2} \rightarrow {#3}\right]}}

\newcommand\treedel[2]{{#1} \setminus {#2}}

\newcommand\tleaf[2]{{#1}_{\bullet_{#2}}}

\newcommand\maxord{n}
\newcommand\midord{k}
\newcommand\salphabet{\Sigma}
\newcommand\cha{a}
\newcommand\chb{b}
\newcommand\chc{c}
\newcommand\chd{c}
\newcommand\stack{s}
\newcommand\sop{\sigma}

\newcommand\stacks[2]{\mathcal{S}^{#2}_{#1}}

\newcommand\kstack[2]{\left[#2\right]_{#1}}

\newcommand\scomp[1]{:_{#1}}

\newcommand\stackops[2]{\mathrm{Ops}^{#2}_{#1}}

\newcommand\srew[2]{\mathrm{rew}_{#1 \rightarrow #2}}

\newcommand\spush[1]{\mathrm{push}_{#1}}

\newcommand\spop[1]{\mathrm{pop}_{#1}}

\newcommand\scollapse[1]{\mathrm{collapse}_{#1}}

\newcommand\scpush[1]{\mathrm{push}^{#1}_1}

\newcommand\akstacks[3]{\mathcal{S}^{#3}_{#1, #2}}

\newcommand\annot[2]{#1^{#2}}

\newcommand\sopen[1]{[_{#1}}

\newcommand\sclose[1]{]_{#1}}

\newcommand\bnode[1]{\rnode{#1}{$\bullet$}}

\newcommand\stacktrees[2]{\mathrm{STrees}^{#1}_{#2}}

\newcommand\strule{\theta}

\newcommand\stsop[3]{{#1} \xrightarrow{#2} {#3}}

\newcommand\stpush[2]{{#1} \xrightarrow{+} \brac{#2}}

\newcommand\stpop[2]{\brac{#1} \xrightarrow{-} {#2}}

\newcommand\stacktreeops[3]{\mathrm{STOps}^{#2, #3}_{#1}}

\newcommand\apstop[3]{\ap{\mathrm{Ap}}{#1, #2, #3}}

\newcommand\gstrs{\mathcal{G}}

\newcommand\rules{\mathcal{R}}

\newcommand\tran{\rightarrow}
\newcommand\reaches{\rightarrow^\ast}

\newcommand\config[2]{\tup{{#1}, {#2}}}

\newcommand\gtrule[3]{{#1} \xrightarrow{#2} {#3}}

\newcommand\lifespan{\iota}

\newcommand\treestacklab{\treelabelling_s}

\newcommand\controls{\mathbb{P}}

\newcommand\control{p}

\newcommand\run{\rho}

\newcommand\ta{\mathcal{T}}
\newcommand\tastate{q}
\newcommand\sastate{r}
\newcommand\tastates{\mathbb{Q}}
\newcommand\sastates{\mathbb{R}}
\newcommand\tadelta{\Delta}
\newcommand\sadelta{\Delta}
\newcommand\tafinals{\mathbb{F}}
\newcommand\safinals{\mathbb{F}}
\newcommand\tastatef{q_f}
\newcommand\sastateset{R}
\newcommand\tastateset{Q}
\newcommand\branch{\mathrm{br}}
\newcommand\satfn{\mathcal{F}}

\newcommand\smodels{\models}
\newcommand\tastateseq{\tilde{\tastate}}
\newcommand\tenv{\mathcal{V}}
\newcommand\tenvempty{\emptyset}
\newcommand\tatrant{\delta}

\newcommand{\Marked}{\mathrm{Marked}}

\newcommand\tatran[5]{{#1} \leftarrow_{{#2}/{#3}} \tup{{#4}, {#5}}}

\newcommand\satran[1]{\xrightarrow{#1}}

\newcommand\satrancol[2]{\xrightarrow[{#2}]{#1}}

\newcommand\slang[2]{\ap{\lang_{#1}}{#2}}

\newcommand\tlang[2]{\ap{\lang_{#1}}{#2}}

\newcommand\tweaklang[1]{\ap{\lang_W}{#1}}

\newcommand\satranfull[4]{{#1} \xrightarrow[{#3}]{#2} \brac{{#4}}}

\newcommand\satranfullk[3]{{#1} \xrightarrow{#2} \brac{{#3}}}

\newcommand\tatranfull[7]{{#1} \leftarrow_{{#2}/{#3}} \tup{{#4}, {#5}, {#6}, {#7}}}

\newcommand\tatranfullk[6]{{#1} \leftarrow_{{#2}/{#3}} \tup{{#4}, {#5}, {#6}}}

\newcommand\prestar[2]{\ap{\mathrm{Pre}^\ast_{#1}}{#2}}

\newcommand\tmodels[1]{\models_{#1}}

\newcommand\tmodelsm[3]{\models_{\envmod{#1}{#2}{#3}}}

\newcommand\treeplus[1]{+_{#1}}

\newcommand\envmod[3]{{#1}\mathord{\left[{#2} \rightarrow {#3}\right]}}

\newcommand\comp[1]{\overline{#1}}

\newcommand\globals{\mathbb{G}}
\newcommand\gstate{g}
\newcommand\gstateseq{{\tilde{\gstate}}}
\newcommand\gstateseqs{\tilde{\globals}}
\newcommand\tadeltainit{\tadelta_{\text{init}}}

\newcommand\tadeltanoapp{\tadelta_{\text{noapp}}}
\newcommand\tadeltapass{\tadelta_{\text{pass}}}

\newcommand\finals{\mathcal{F}}

\newcommand\hostackops[2]{\mathrm{HOps^{#2}_{#1}}}

\newcommand\hocpush[1]{\mathrm{push}_{#1}}

\newcommand\horule[4]{\tup{#1, #2, #3, #4}}

\newcommand\ccopy[2]{\tup{#1, #2}}

\newcommand\newtastate[1]{\tastate_{#1}}

\newcommand\newsastate[2]{\sastate_{\brac{#1, #2}}}

\usepackage{color}

\ifdraft
    \newcommand\todo[1]{{\color{red} [\textbf{To do:} #1]}}
    \newcommand\mh[1]{{\color{blue} [#1 - \textbf{Matt}]}}
    \newcommand\vp[1]{{\color{orange} [#1 - \textbf{Vincent}]}}

\else
    
    \newcommand\todo[1]{}
    \newcommand\mh[1]{}
    \newcommand\vp[1]{}

\fi

\newtheorem{theorem}{Theorem}[section]
\newtheorem{definition}{Definition}[section]
\newtheorem{proposition}{Proposition}[section]
\newtheorem{lemma}{Lemma}[section] 
\newtheorem{property}{Property}[section] 

\newcommand\reftheorem[1]{\expandafter\csname reftheorem#1\endcsname}

\newenvironment{namedlemma}[2]{%
    \expandafter\gdef\csname reflemma#1\endcsname{%
        Lemma~\ref{#1} (#2)%
    }%
    \begin{lemma}[{#2}] \label{#1}%
}{%
    \end{lemma}%
}
\newcommand\reflemma[1]{\expandafter\csname reflemma#1\endcsname}

\newenvironment{namedproperty}[2]{%
    \expandafter\gdef\csname refproperty#1\endcsname{%
        Property~\ref{#1} (#2)%
    }%
    \begin{property}[{#2}] \label{#1}%
}{%
    \end{property}%
}
\newcommand\refproperty[1]{\expandafter\csname refproperty#1\endcsname}

\newcommand\refdefinition[1]{\expandafter\csname refdefinition#1\endcsname}

\usepackage[T1]{fontenc}
\usepackage{minted}
\usepackage{upquote}
\AtBeginDocument{%
}

\title{Annotated Stack Trees\thanks{This work was
supported by the Engineering and Physical Sciences Research Council
[EP/K009907/1].}}
\author{M. Hague\inst{1} \and V. Penelle\inst{2}}
\institute{Royal Holloway, University of London \and LIGM, Universit\'e
Paris-Est}

\begin{document}

    \maketitle

    \begin{abstract}

Annotated pushdown automata provide an automaton model of higher-order recursion
schemes, which may in turn be used to model higher-order programs for the
purposes of verification.  We study Ground Annotated Stack Tree Rewrite Systems
-- a tree rewrite system where each node is labelled by the configuration of an
annotated pushdown automaton.  This allows the modelling of fork and join
constructs in higher-order programs and is a generalisation of higher-order
stack trees recently introduced by Penelle.

We show that, given a regular set of annotated stack trees, the set of trees
that can reach this set is also regular, and constructible in $\maxord$-EXPTIME
for an order-$\maxord$ system, which is optimal.  We also show that our
    construction can be extended to allow a global state through which unrelated
    nodes of the tree may communicate, provided the number of communications is
    subject to a fixed bound.

    \end{abstract}

\section{Introduction}

Modern day programming increasingly embraces higher-order programming, both via
the inclusion of higher-order constructs in languages such as C++, JavaScript
and Python, but also via the importance of \emph{callbacks} in highly popular
technologies such as jQuery and Node.js.  For example, to read a file in
Node.js, one would write
\begin{minted}{javascript}
    fs.readFile('f.txt', function (err, data) { ..use data.. });
\end{minted}
In this code, the call to \verb+readFile+ spawns a new thread that
asynchronously reads \verb+f.txt+ and sends the \verb+data+ to the function
argument.  This function will have access to, and frequently use, the closure
information of the scope in which it appears.  The rest of the program runs
\emph{in parallel} with this call.  This style of programming is
fundamental to both jQuery and Node.js programming, as well as being a popular
for programs handling input events or slow IO operations such as fetching
remote data or querying databases (e.g. HTML5's indexedDB).

Analysing such programs is a challenge for verification tools which usually do
not model higher-order recursion, or closures, accurately.  However, several
higher-order model-checking tools have been recently developed.  This trend was
pioneered by Kobayashi\etal~\cite{K11} who developed an \emph{intersection type}
technique for analysing \emph{higher-order recursion schemes} -- a model of
higher-order computation.  This was implemented in the \trecs tool~\cite{K09}
which demonstrated the feasibility of higher-order model-checking in practice,
despite the high theoretical complexities ($(\maxord-1)$-EXPTIME for an
order-$\maxord$ recursion scheme).  This success has led to the development of
several new tools for analysing recursion schemes: \gtrecs~\cite{K11b,gtrecs2},
\travmc~\cite{NRO12}, \cshore~\cite{BCHS13}, \horsat~\cite{BK13}, and
\preface~\cite{RNO14}.

In particular, the \cshore tool is based on an automata model of recursion
schemes called \emph{annotated (or collapsible) pushdown systems}~\cite{HMOS08}.
This is a generalisation of pushdown systems -- which accurately model
first-order recursion -- to the higher-order case.  \cshore implements a
\emph{saturation} algorithm to perform a backwards reachability analysis, which
first appeared in ICALP 2012~\cite{BCHS12}.  Saturation was popularised by
Bouajjani\etal~\cite{BEM97} for the analysis of pushdown systems, which was
implemented in the successful Moped tool~\cite{S02,SBSE07}.

\paragraph*{Contributions}
In this work we introduce a generalisation of annotated pushdown systems:
\emph{ground annotated stack tree rewrite systems (GASTRS)}.  A configuration of
a GASTRS is an \emph{annotated stack tree} -- that is, a tree where each node is
labelled by the configuration of an annotated pushdown system.  Operations may
update the leaf nodes of the tree, either by updating the configuration,
creating new leaf nodes, or destroying them.  Nodes are
created and destroyed using
\[
    \stpush{\control}{\control_1, \ldots, \control_\numof}
    \text{  and  }
    \stpop{\control'_1, \ldots, \control'_\numof}{\control'}
\]
which can be seen as spawning $\numof$ copies of the current process (including
closure information) using the first rule, and then later joining these
processes with the second rule, returning control to the previous execution
(parent node).  Alternatively, we can just use
$\stpush{\control}{\control_1, \control_2}$
for a basic fork that does not join.

This model is a generalisation of \emph{higher-order stack trees} recently
introduced by Penelle~\cite{P15}, where the tree nodes are labelled by a
restriction of annotated pushdown automata called \emph{higher-order pushdown
automata}.

As our main contribution, we show that the global backwards reachability problem
for GASTRSs can be solved via a saturation technique.  That is, given a regular
target set of annotated stack trees, we compute a regular representation of
all trees from which there is a run of the system to the target set.  Note
that being able to specify a target set of trees allows us to identify error
states such as race conditions between threads.  Our result is a
generalisation of the ICALP 2012 algorithm, and as such, may be implemented
as part of the \cshore tool.

Moreover, we define a notion of regularity amenable to saturation which is also
closed under the standard boolean operations.

As a final contribution, we show that the model can be extended to allow a
bounded amount of communication between separate nodes of the tree.  I.e., we
add a global state to the system and perform a ``context-bounded''
analysis~\cite{QR05}, where the global state can only be changed an \textit{a
priori} fixed number of times.

\paragraph{Related Work}

Annotated pushdown systems are a generalisation of higher-order pushdown systems
that provide a model of recursion schemes subject to a technical constraint
called \emph{safety}~\cite{M76,KNU02} and are closely related to the Caucal
hierarchy~\cite{CW03}.  Parys has shown that safety is a genuine constraint on
definable traces~\cite{P11}.  Panic automata provided the first model of
order-$2$ schemes, while annotated pushdown systems model schemes of arbitrary
order.  These formalisms have good model-checking properties.  E.g.
$\mu$-calculus decidability~\cite{O06,HMOS08}.  Krivine machines can also be
used to model recursion schemes~\cite{SW11}.

There has been some work studying concurrent variants of recursion scheme model
checking, including a context-bounded algorithm for recursion
schemes~\cite{KI13}, and further underapproximation methods such as
phase-bounded, ordered, and scope-bounding~\cite{H13,S09}.  These works
allow only a fixed number of threads.

Dynamic thread creation is permitted by both Yasukata\etal~\cite{YKM14} and by
Chadha and Viswanathan~\cite{CV07}.  In Yasukata\etal's model, recursion schemes
may spawn and join threads.  Communication is permitted only via nested locks,
whereas in our model we allow shared memory, but only a bounded number of memory
updates.  Their work is a generalisation of results for order-1 pushdown
systems~\cite{GLMSW11}.  Chadha and Viswanathan allow threads to be spawned, but
only one thread runs at a time, and must run to completion. Moreover, the tree
structure is not maintained.

Saturation methods also exist for \emph{ground tree rewrite systems} and related
systems~\cite{L03,B69,LS98}, though use different techniques.  Our
context-bounded model relates to weak GTRS with state introduced by
Lin~\cite{L12}.  Adding such weak state to process rewrite systems was
considered by Kret{\'{\i}}nsk{\'{y}\etal~\cite{KRS04}.

A saturation technique has also been developed for dynamic trees of pushdown
processes~\cite{BMT05}.  These are trees where each process on each node is
active (in our model, only the leaf nodes are active).  However, their spawn
operations do not copy the current process, losing closure information.  It
would be interesting and non-trivial to study the combination of both
approaches.

Penelle proves decidability of first order logic with reachabilty
over rewriting graphs of ground stack tree rewriting systems~\cite{P15}.
This may be used for a context-bounded reachability result for
higher-order stack trees.  This result relies on MSO decidability over the
configuration graphs of higher-order pushdown automata, through a finite set
interpretation of any rewriting graph of a ground stack tree rewriting system
into a configuration graph of a higher pushdown automaton. This does not hold
for annotated pushdown automata.

\section{Preliminaries}

\subsubsection*{Trees}
An ordered tree over arity at most $\maxarity$ over a set of labels $\treelabels$
is a tuple
$\tup{\treedom, \treelabelling}$
where
$\treedom \subset \set{1,\ldots,\maxarity}^\ast$
is a tree domain such that
$\tnode \tdiri \in \treedom$
implies
$\tnode \in \treedom$
(prefix closed), and
$\tnode \tdirj \in \treedom$
for all
$\tdirj < \tdiri$
(younger-sibling closed), and
$\treelabelling : \treedom \rightarrow \treelabels$
is a labelling of the nodes of the tree.
\easyicalp{
    Let
    $\tnode \treeanc \tnode'$
    denote that $\tnode$ is an ancestor (inclusive) of $\tnode'$ in the tree.
}{}
We write
$\treemod{\tree}{\tnode}{\treelab}$
to denote the tree
$\tree' = \tup{\treedom \cup \set{\tnode}, \treelabelling'}$
where
$\ap{\treelabelling'}{\tnode} = \treelab$
and
$\ap{\treelabelling'}{\tnode'} = \ap{\treelabelling}{\tnode'}$
for
$\tnode' \neq \tnode$,
whenever
$\tree = \tup{\treedom, \treelabelling}$
and
$\treedom \cup \set{\tnode}$
is a valid tree domain.  We will also write
$\tree' = \treedel{\tree}{\tnodeset}$
to denote the tree obtained by removing all subtrees rooted at $\tnode \in
\tnodeset$ from $\tree$.
\easyicalp{
    That is
    $\tree' = \tup{\treedom', \treelabelling'}$
    when
    $\tree = \tup{\treedom, \treelabelling}$
    and
    \[
        \begin{array}{rcl}
            \treedom'
            &=&
            \treedom \setminus \setcomp{\tnode'}
                                       {\tnode \in \tnodeset \land
                                        \tnode \treeanc \tnode'} \\

            \ap{\treelabelling'}{\tnode}
            &=&
            \begin{cases}
                \ap{\treelabelling}{\tnode} & \tnode \in \treedom' \\
                \text{undefined} & \text{otherwise.}
            \end{cases}
        \end{array}
    \]
}{}

\subsubsection*{Annotated stacks}

Let $\salphabet$ be a set of stack symbols.  An annotated stack of
order-$\maxord$ is an order-$\maxord$ stack in which stack symbols are
annotated with stacks of order at most $\maxord$.  For the rest of the paper, we
fix the maximal order to $\maxord$, and use $\midord$ to range between $\maxord$
and $1$.  We simultaneously define for all
$1 \leq \midord \leq \maxord$,
the set
$\akstacks{\midord}{\maxord}{\salphabet}$
of stacks of order-$\midord$ whose symbols are annotated by stacks of order at
most $\maxord$.  Note, we use subscripts to indicate the order of a stack.
We ensure all stacks are finite by using the least fixed-point.  When the
maximal order $\maxord$ is clear, we write
$\stacks{\midord}{\salphabet}$
instead of
$\akstacks{\midord}{\maxord}{\salphabet}$.

\begin{definition}[Annotated Stacks]
    The family of sets
    $\brac{\akstacks{\midord}
                  {\maxord}
                  {\salphabet}}_{1 \leq \midord \leq \maxord}$
    is the smallest family (for point-wise inclusion) such that:
    \begin{itemize}
        \item
            for all
            $2 \leq \midord \leq \maxord$,
            $\akstacks{\midord}{\maxord}{\salphabet}$
            is the set of all (possibly empty)
            sequences
            $\kstack{\midord}{\stack_1 \ldots \stack_\numof}$
            with
            $\stack_1, \ldots, \stack_\numof
                \in \akstacks{\midord-1}{\maxord}{\salphabet}$.

        \item
            $\akstacks{1}{\maxord}{\salphabet}$
            is all sequences
            $\kstack{1}{
                \annot{\cha_1}{\stack_1}
                \ldots
                \annot{\cha_\numof}{\stack_\numof}}$
            with
            $\numof \geq 0$
            and for all
            $1 \leq \idxi \leq \numof$, j
            $\cha_\idxi$ is a stack symbol in $\salphabet$ and $\stack_\idxi$ is
            an annotated stack in
            $\bigcup\limits_{1 \leq \midord \leq \maxord}
                \akstacks{\midord}{\maxord}{\salphabet}$.
    \end{itemize}
\end{definition}

We write
$\stack \scomp{\midord} \stack'$
--- where $\stack$ is order-$(\midord-1)$ ---
to denote the stack obtained by
placing $\stack$ on top of $\stack'$.  That is,
\begin{itemize}
    \item
        if
        $\stack' = \kstack{\midord}{\stack_1 \ldots \stack_\numof}$
        then
        $\stack \scomp{\midord} \stack'
            = \kstack{\midord}
                     {\stack \stack_1 \ldots \stack_\numof}$, and

    \item
        if
        $\stack' = \kstack{\midord'}{\stack_1 \ldots \stack_\numof}$
        with
        $\midord' > \midord$
        then
        $\stack \scomp{\midord} \stack'
            = \kstack{\midord'}
                     {\brac{\stack \scomp{\midord} \stack_1}
                      \stack_2
                      \ldots
                      \stack_\numof}$.
\end{itemize}

This composition associates to
the right.  For example, the order-$3$ stack
$\kstack{3}{\kstack{2}{\kstack{1}{\annot{\cha}{\stack} \chb}}}$
can be written
$\stack_1 \scomp{3} \stack_2$
where $\stack_1$ is the order-$2$ stack
$\kstack{2}{\kstack{1}{\annot{\cha}{\stack} \chb}}$
and
$\stack_2$ is the empty order-$3$ stack
$\kstack{3}{}$.
Then
$\stack_1 \scomp{3} \stack_1 \scomp{3} \stack_2$
is
$\kstack{3}{\kstack{2}{\kstack{1}{\annot{\cha}{\stack} \chb}}
            \kstack{2}{\kstack{1}{\annot{\cha}{\stack} \chb}}}$.

Note that we cannot write
$\brac{\stack_1 \scomp{\midord} \stack_2} \scomp{\midord} \stack_3$
since
$\brac{\stack_1 \scomp{\midord} \stack_2}$
is not order-$(\midord-1)$.

\subsubsection*{Operations on Order-$\maxord$ Annotated Stacks}

For a given alphabet $\salphabet$, we define the set
$\stackops{\maxord}{\salphabet}$ of stack operations inductively as follows:
\[
    \begin{array}{c}
        \stackops{0}{\salphabet}
        =
        \setcomp{\srew{\cha}{\chb}}{\cha, \chb \in \salphabet}
        \qquad
        \stackops{1}{\salphabet}
        =
        \set{\scpush{1}, \spop{1}}
        \cup
        \stackops{0}{\salphabet}
        \\
        \stackops{\maxord}{\salphabet}
        =
        \set{\scpush{\maxord},
             \spush{\maxord},
             \spop{\maxord},
             \scollapse{\maxord}}
        \cup
        \stackops{(\maxord-1)}{\salphabet}
    \end{array}
\]

We define each operation for a stack $\stack$.
Annotations are created by
$\scpush{\midord}$,
which adds a character to the top of a stack
$\stack \scomp{(\midord+1)} \stack'$
annotated by
$\ap{\spop{\midord}}{\stack}$.
This gives the new character access to the context in which it was created.
\begin{enumerate}
    \item
        We set
        $\ap{\srew{\cha}{\chb}}
            {\annot{\cha}{\stack'} \scomp{1} \stack}
        = \annot{\chb}{\stack'} \scomp{1} \stack$.

    \item
        We set
        $\ap{\scpush{\midord}}{\stack}
            = \annot{\cha}{\stack_\midord}
              \scomp{1}
              \stack$
        when
        $\stack = \annot{\cha}{\stack_1}
                  \scomp{1}
                  \stack_2
                  \scomp{2}
                  \cdots
                  \scomp{\midord}
                  \stack_\midord
                  \scomp{(\midord+1)}
                  \cdots
                  \scomp{\maxord}
                  \stack_\maxord$.

    \item
        We set
        $\ap{\spush{\midord}}{\stack \scomp{\midord} \stack'}
            = \stack \scomp{\midord} \stack \scomp{\midord} \stack'$.

    \item
        We set
        $\ap{\spop{\midord}}{\stack \scomp{\midord} \stack'} = \stack'$.

    \item
        We set
        $\ap{\scollapse{\midord}}{\annot{\cha}{\stack}
                                  \scomp{1}
                                  \stack_1 \scomp{(\midord+1)} \stack_2}
            = \stack \scomp{(\midord+1)} \stack_2$
        when $\stack$ is order-$\midord$ and
        $\maxord > \midord \geq 1$;
        and
        $\ap{\scollapse{\maxord}}{\annot{\cha}{\stack} \scomp{1} \stack'}
            =  \stack$
        when $\stack$ is order-$\maxord$.
\end{enumerate}

\section{Annotated Stack Trees}

An \emph{annotated stack tree} is a tree whose nodes are labelled by annotated
stacks.  Furthermore, each leaf node is also labelled with a control state.  Let
$\stacktrees{\maxord}{\salphabet}$
denote the set of order-$\maxord$ annotated stack trees over $\salphabet$.

\begin{definition}[Order-$\maxord$ Annotated Stack Trees]
    An \emph{order-$\maxord$ annotated stack tree} over an alphabet $\salphabet$
    and set of control states $\controls$ is a
    $\brac{
        \stacks{\maxord}{\salphabet} \cup
        \brac{
            \controls \times \stacks{\maxord}{\salphabet}
        }
    }$-labelled
    tree
    $\tree = \tup{\treedom, \treelabelling}$
    such that for all leaves $\tnode$ of $\tree$ we have
    $\ap{\treelabelling}{\tnode}
        \in \controls \times \stacks{\maxord}{\salphabet}$
    and for all internal nodes $\tnode$ of $\tree$ we have
    $\ap{\treelabelling}{\tnode} \in \stacks{\maxord}{\salphabet}$.
\end{definition}

\subsection{Annotated Stack Tree Operations}

\begin{definition}[Order-$\maxord$ Annotated Stack Tree Operations]
    Over a given finite alphabet $\salphabet$ and finite set of control states
    $\controls$, the set of \emph{order-$\maxord$ stack tree operations}
    is defined to be
    \[
        \begin{array}{rcl}
            \stacktreeops{\maxord}{\salphabet}{\controls}
            &=&
            \setcomp{
                \stpush{\control}{\control_1, \ldots, \control_\numof},
                \stpop{\control_1, \ldots, \control_\numof}{\control}
            }{
                \control, \control_1, \ldots, \control_\numof \in \controls
            } \cup
            \\
            & &
            \setcomp{\stsop{\control}{\sop}{\control'}}
                    {\sop \in \stackops{\maxord}{\salphabet} \land
                     \control, \control' \in \controls} \ .
        \end{array}
    \]
\end{definition}
Stack operations may be applied to any leaf of the tree.  Let
$\tleaf{\tree}{\idxi}$
denote the $\idxi$th leaf of tree $\tree$.  We define the local application of a
operation to the $\idxi$th leaf as follows.  Let
$\tree = \tup{\treedom, \treelabelling}$
and
$\ap{\treelabelling}{\tleaf{\tree}{\idxi}} = \tup{\control, \stack}$
\[
    \begin{array}{rcl}
        \apstop{\stsop{\control}{\sop}{\control'}}{\idxi}{\tree}
        &=&
        \treemod{\tree}
                {\tleaf{\tree}{\idxi}}
                {\tup{\control', \ap{\sop}{\stack}}}
        \\

        \apstop{\stpush{\control}{\control_1, \ldots, \control_\numof}}{\idxi}{\tree}
        &=&
        \treemod{\treemod{\treemod{\tree}
                                  {\tleaf{\tree}{\idxi}}
                                  {\stack}}
                         {\tleaf{\tree}{\idxi}1}
                         {\tup{\control_1, \stack}}
                 \cdots}
                {\tleaf{\tree}{\idxi}\numof}
                {\tup{\control_\numof, \stack}}
    \end{array}
\]
and when
$\tleaf{\tree}{\idxi} = \tnode 1$,
\ldots,
$\tleaf{\tree}{\idxi+\numof-1} = \tnode \numof$
are the only children of $\tnode$,
$\ap{\treelabelling}{\tleaf{\tree}{\idxi}} = \tup{\control_1, \stack_1}$,
\ldots,
$\ap{\treelabelling} {\tleaf{\tree}{\idxi+{\numof-1}}}
    = \tup{\control_\numof, \stack_\numof}$,
and
$\ap{\treelabelling}{\tnode} = \stack$,
\[
    \apstop{\stpop{\control_1, \ldots, \control_\numof}{\control}}{\idxi}{\tree}
    =
    \treemod{\brac{
                 \treedel{\tree}
                         {\set{\tleaf{\tree}{\idxi},
                               \ldots,
                               \tleaf{\tree}{\idxi+{\numof-1}}}}
             }}
            {\tnode}
            {\tup{\control, \stack}} \ .
\]

For all
$\strule \in \stacktreeops{\maxord}{\salphabet}{\controls}$
we write
$\ap{\strule}{\tree}$
to denote the set
$\setcomp{\tree'}{\exists \idxi .  \tree' = \apstop{\strule}{\idxi}{\tree}}$.

\subsection{Ground Annotated Stack Tree Rewrite Systems}

\begin{definition}[Order-$\maxord$ Ground Annotatee Stack Tree Rewrite Systems]
    An \emph{order-$\maxord$ ground annotated stack tree rewrite system (GASTRS)}
    $\gstrs$ is a tuple
    $\tup{\salphabet, \controls, \rules}$
    where $\salphabet$ is a finite stack alphabet, $\controls$ is a finite set
    of control states, and
    $\rules \subset \stacktreeops{\maxord}{\salphabet}{\controls}$
    is a finite set of operations.
\end{definition}

A configuration of an order-$\maxord$ GASTRS is an order-$\maxord$ annotated
stack tree $\tree$ over alphabet $\salphabet$.  We have a transition
$\tree \tran \tree'$
whenever there is some
$\strule \in \rules$
and
$\tree' \in \ap{\strule}{\tree}$.
We write
$\tree \reaches \tree'$ when there is a run
$\tree = \tree_0 \tran \cdots \tran \tree_\numof = \tree'$.

\subsection{Regular Sets of Annotated Stack Trees}

We define a notion of annotated stack tree automata for recognising regular sets
of annotated stack trees.  We give an initial exposition here, with more details
(definitions and proofs) in Appendix~\ref{sec:aut-particulars}.
In particular,
\easyicalp{
    we have the following result.

    \begin{proposition}
        Annotated stack tree automata form an effective boolean algebra, membership
        is in linear time, and emptiness is PSPACE-complete.
    \end{proposition}
}{
    stack tree automata form an effective boolean algebra, membership is linear
    time, and emptiness is PSPACE-complete.
}
Transitions of stack tree automata are labelled by states of stack automata which have a
further nested structure~\cite{BCHS12}.  These automata are based on a similar
automata model by Bouajjani and Meyer~\cite{BM04}.  We give the formal
definition with intuition following.

\begin{definition}[Order-$\maxord$ Annotated Stack Tree Automata]
    An \emph{order-$\maxord$ stack tree automaton} over a given stack alphabet
    $\salphabet$ and set of control states $\controls$ is a tuple
    \[
        \ta = \tup{
                   \tastates, \sastates_\maxord,\ldots,\sastates_1,
                   \salphabet,
                   \tadelta, \sadelta_\maxord,\ldots,\sadelta_1,
                   \controls,
                   \tafinals, \safinals_\maxord,\ldots,\safinals_1
               }
    \]
    where
    $\salphabet$ is a finite stack alphabet,
    $\tastates$ is a finite set of states,
    \[
        \tadelta \subset \tastates
                         \times
                         \setcomp{\tup{\idxi, \numof}}{1 \leq \idxi \leq \numof}
                         \times
                         \brac{\tastates \setminus \tafinals}
                         \times
                         \sastates_\maxord
    \]
    is a finite set of transitions,
    $\controls \subseteq \tastates$ and $\tafinals \subseteq \tastates$ are initial and
    final states respectively,
    and
    \begin{enumerate}
        \item
            for all
            $\maxord \geq \midord \geq 2$,
            we have
            $\sastates_\midord$ is a finite set of states,
            $\sadelta_\midord \subseteq \sastates_\midord
                                        \times
                                        \sastates_{\midord-1}
                                        \times
                                        2^{\sastates_\midord}$
            is a transition relation, and
            $\safinals_\midord \subseteq \sastates_\midord$
            is a set of accepting states, and

         \item
            $\sastates_1$ is a finite set of states,
            $\sadelta_1 \subseteq \bigcup\limits_{2 \leq \midord \leq \maxord}
                                  \brac{\sastates_1
                                        \times \salphabet
                                        \times
                                        2^{\sastates_\midord}
                                        \times
                                        2^{\sastates_1}}$
            is a transition relation, and
            $\safinals_1 \subseteq \sastates_1$
            is a set of accepting states.
    \end{enumerate}
\end{definition}

\subsubsection{Accepting Stacks}

Order-$\midord$ stacks are recognised from states in $\sastates_\midord$.  A
transition
$\tup{\sastate, \sastate', \sastateset} \in \sadelta_\midord$
from $\sastate$ to $\sastateset$ for some $\midord > 1$
is denoted
$\sastate \satran{\sastate'} \sastateset$
and can be fired when the stack is
$\stack \scomp{\midord} \stack'$
and $\stack$ is accepted from
$\sastate' \in \sastates_{(\midord-1)}$.
The remainder of the stack $\stack'$ must be accepted from all states in
$\sastateset$.  At order-$1$, a transition
$\tup{\sastate, \cha, \sastateset_\branch, \sastateset} \in \sadelta_1$
is denoted
$\sastate \satrancol{\cha}{\sastateset_\branch} \sastateset$
and is a standard alternating $\cha$-transition with the additional
requirement that the stack annotating $\cha$ is accepted from all states in
$\sastateset_\branch$.  A stack is accepted if a subset of $\safinals_\midord$
is reached at the end of each order-$\midord$ stack.  Note, we give a more
formal definition of a run in
Appendix~\ref{sec:aut-particulars}.
We write
$\stack \in \slang{\sastate}{\ta}$
whenever $\stack$ is accepted from a state $\sastate$.

An order-$\maxord$ stack can be represented naturally as an edge-labelled tree
over the alphabet
$\set{\sopen{\maxord-1},\ldots,\sopen{1},
      \sclose{1},\ldots,\sclose{\maxord-1}}
 \uplus
 \salphabet$,
with $\salphabet$-labelled edges having a second target to the
tree representing the annotation.  For technical convenience, a tree
representing an order-$\midord$ stack does not use $\sopen{\midord}$ or
$\sclose{\midord}$ symbols (these appear uniquely at the beginning and end of the
stack). An example order-$3$ stack is given below, with only a few annotations
shown.  The annotations are order-$3$ and order-$2$ respectively.
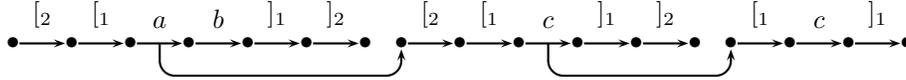
\begin{center}
    \vspace{4ex}
    \begin{psmatrix}[nodealign=true,colsep=2ex,rowsep=2ex]
        \bnode{N1} && \bnode{N2} && \bnode{N3} &\pnode{N34}& \bnode{N4} &&
        \bnode{N5} && \bnode{N6} && \bnode{N7} &

        \bnode{N8} && \bnode{N9} && \bnode{N10} &\pnode{N1011}& \bnode{N11} &&
        \bnode{N12} && \bnode{N13} &

        \bnode{N14} &\pnode{N1415}& \bnode{N15} && \bnode{N16} && \bnode{N17} \\

        \psset{angle=-90,linearc=.2}
        \ncline{->}{N1}{N2}^{$\sopen{2}$}
        \ncline{->}{N2}{N3}^{$\sopen{1}$}
        \ncline{->}{N3}{N4}^{$\cha$}
        \ncbar{->}{N34}{N8}
        \ncline{->}{N4}{N5}^{$\chb$}
        \ncline{->}{N5}{N6}^{$\sclose{1}$}
        \ncline{->}{N6}{N7}^{$\sclose{2}$}

        \ncline{->}{N8}{N9}^{$\sopen{2}$}
        \ncline{->}{N9}{N10}^{$\sopen{1}$}
        \ncline{->}{N10}{N11}^{$\chc$}
        \ncbar{->}{N1011}{N14}
        \ncline{->}{N11}{N12}^{$\sclose{1}$}
        \ncline{->}{N12}{N13}^{$\sclose{2}$}

        \ncline{->}{N14}{N15}^{$\sopen{1}$}
        \ncline{->}{N15}{N16}^{$\chd$}
        \ncline{->}{N16}{N17}^{$\sclose{1}$}
    \end{psmatrix}
\end{center}

An example (partial) run over this stack is pictured below, using transitions
$\sastate_3 \satran{\sastate_2} \sastateset_3 \in \sadelta_3$,
$\sastate_2 \satran{\sastate_1} \sastateset_2 \in \sadelta_2$,
and
$\sastate_1 \satrancol{\cha}{\sastateset_\branch} \sastateset_1 \in \sadelta_1$.
The node labelled $\sastateset_\branch$ begins a run on the stack annotating $\cha$.
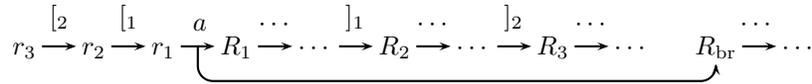
\begin{center}
    \vspace{2ex}
    \begin{psmatrix}[nodealign=true,colsep=2ex,rowsep=1.25ex]
        \Rnode{N1}{$\sastate_3$} &             & \Rnode{N2}{$\sastate_2$}         &  &
        \Rnode{N3}{$\sastate_1$} & \pnode{N34} & \Rnode{N4}{$\sastateset_1$}      &  &
        \Rnode{N5}{$\cdots$}     &             & \Rnode{N6}{$\sastateset_2$}      &  &
        \Rnode{N7}{$\cdots$}     &             & \Rnode{N8}{$\sastateset_3$}      &  &
        \Rnode{N9}{$\cdots$}     &             & \Rnode{N10}{$\sastateset_\branch$} &  &
        \Rnode{N11}{$\cdots$} \\

        \psset{nodesep=.5ex,angle=-90,linearc=.2}
        \ncline{->}{N1}{N2}^{$\sopen{2}$}
        \ncline{->}{N2}{N3}^{$\sopen{1}$}
        \ncline{->}{N3}{N4}^{$\cha$}
        \ncbar[arm=1.5ex,nodesepA=0]{->}{N34}{N10}
        \ncline{->}{N4}{N5}^{$\cdots$}
        \ncline{->}{N5}{N6}^{$\sclose{1}$}
        \ncline{->}{N6}{N7}^{$\cdots$}
        \ncline{->}{N7}{N8}^{$\sclose{2}$}
        \ncline{->}{N8}{N9}^{$\cdots$}

        \ncline{->}{N10}{N11}^{$\cdots$}
    \end{psmatrix}
\end{center}

\subsubsection{Accepting Stack Trees}

Annotated stack tree automata are bottom-up tree automata whose transitions are
labelled by states from which stacks are accepted.  We denote by
\[
    \tatran{\tastate}{\idxi}{\numof}{\tastate'}{\sastate}
\]
a transition
$\tup{\tastate, \idxi, \numof, \tastate', \sastate} \in \tadelta$.
Observe that
$\tastate' \notin \tafinals$
by definition.
When a node $\tnode$ has children
$\tnode_1, \ldots, \tnode_\numof$,
the transition above could be applied to the $\idxi$th child $\tnode_\idxi$.  It
can be applied when $\tnode_\idxi$ is already labelled by $\tastate'$ and the
stack $\stack_\idxi$ attached to $\tnode_\idxi$ is accepted from state $\sastate$
of the stack automaton.  If it is applied, then $\tastate$ will be set as the
label of the parent $\tnode$.  Over runs of the automaton we enforce that every
child is present and the transitions applied at each child agree on the state
assigned to its parent.

Let
$\ap{\treestacklab}{\tnode} = \stack$
when
$\ap{\treelabelling}{\tnode} = \tup{\control, \stack}$
or
$\ap{\treelabelling}{\tnode} = \stack$.
Given an order-$\maxord$ annotated stack tree
$\tree = \tup{\treedom, \treelabelling}$
a run of an automaton $\ta$ is a $\tastates$-labelled tree
$\tup{\treedom, \treelabelling'}$
where each leaf $\tnode$ of $\tree$ has
$\ap{\treelabelling'}{\tnode} = \control$
whenever
$\ap{\treelabelling}{\tnode} = \tup{\control, \stack}$
for some $\stack$, and each internal node $\tnode$ with children
$\tnode 1, \ldots, \tnode \numof$
has a label
$\ap{\treelabelling'}{\tnode} = \tastate$
only if we have transitions
\[
    \tatran{\tastate}{1}{\numof}{\tastate_1}{\sastate_1},
    \ldots,
    \tatran{\tastate}{\numof}{\numof}{\tastate_\numof}{\sastate_\numof},
\]
and
$\ap{\treelabelling'}{\tnode \idxi} = \tastate_\idxi$
and
$\ap{\treestacklab}{\tnode \idxi} \in \slang{\sastate_\idxi}{\ta}$
for all
$1 \leq \idxi \leq \numof$.
Finally
$\ap{\treelabelling'}{\troot} = \tastate$
and we have a transition
$\tatran{\tastatef}{1}{1}{\tastate}{\sastate}$
with
$\tastatef \in \tafinals$
and
$\ap{\treestacklab}{\troot} \in \slang{\sastate}{\ta}$.

We write
$\langof{\ta}$
to denote the set of trees accepted by $\ta$.

\subsection{Notation and Conventions}
\label{ssec:notations}

\subsubsection{Number of Transitions}
We assume for all pairs of states
$\tastate, \tastate' \in \tastates$
and each $\idxi, \numof$ there is at most one transition of the form
$\tatran{\tastate}{\idxi}{\numof}{\tastate'}{\sastate}$.
Similarly we assume for all
$\sastate \in \sastates_\midord$
and
$\sastateset \subseteq \sastates_\midord$
that there is at most one transition of the form
$\sastate \satran{\sastate'} \sastateset \in \sadelta_\midord$.
This condition can easily
be ensured by replacing pairs of transitions
$\sastate \satran{\sastate_1} \sastateset$
and
$\sastate \satran{\sastate_2} \sastateset$
with a single transition
$\sastate \satran{\sastate'} \sastateset$,
where $\sastate'$ accepts the union of the languages of stacks accepted from
$\sastate_1$ and $\sastate_2$.  Similarly for transitions in $\tadelta$.

\subsubsection{Short-form Notation}

Consider the example run shown above.  This run reads the top of every level of
the stack: the transition to $\sastateset_3$ reads the topmost order-$2$ stack,
the transition to $\sastateset_2$ reads the order-$1$ stack at the top of this
stack, and the transition to $\sastateset_1$ and $\sastateset_\branch$ reads the
top character of the order-$1$ stack.

The saturation algorithm relies on stack updates only affecting the topmost
part of the stack.  Thus, we need a notation for talking about the beginning of
the run.  Hence, we will write the run in the figure above (that reads the
topmost parts of the stack) as a ``short-form'' transition
\[
    \satranfull{\sastate_3}
               {\cha}
               {\sastateset_\branch}
               {\sastateset_1, \ldots, \sastateset_3} \ .
\]
In the following, we define this notation formally, and generalise it to
transitions of a stack tree automaton.  In general, we write
\[
    \satranfull{\sastate}
               {\cha}
               {\sastateset_\branch}
               {\sastateset_1,\ldots,\sastateset_\midord}
    \text{    and    }
    \satranfullk{\sastate}
                {\sastate'}
                {\sastateset_{\midord'+1}, \ldots, \sastateset_\midord} .
\]
In the first case,
$\sastate \in \sastates_\midord$
and there exist
$\sastate_{\midord-1}, \ldots, \sastate_1$
such that
$\sastate
 \satran{\sastate_{\midord-1}}
 \sastateset_\midord \in \sadelta_\midord$,
$\sastate_{\midord-1}
 \satran{\sastate_{\midord-2}}
 \sastateset_{\midord-1} \in \sadelta_{\midord-1}$,
\ldots,
$\sastate_1
 \satrancol{\cha}{\sastateset_\branch}
 \sastateset_1 \in \sadelta_1$.
Since we assume at most one transition between any state and set of states, the
intermediate states
$\sastate_{\midord-1}, \ldots, \sastate_1$
are uniquely determined by
$\sastate, \cha, \sastateset_\branch$
and
$\sastateset_1, \ldots, \sastateset_\midord$.

In the second case, either
$\midord = \midord'$
and
$\sastate = \sastate' \in \sastates_\midord$,
or
$\midord > \midord'$
and we have
$\sastate \in \sastates_\midord$,
$\sastate' \in \sastates_{\midord'}$,
and there exist
$\sastate_{\midord-1}, \ldots, \sastate_{\midord'+1}$
with
$\sastate
 \satran{\sastate_{\midord-1}}
 \sastateset_\midord \in \sadelta_\midord$,
$\sastate_{\midord-1}
 \satran{\sastate_{\midord-2}}
 \sastateset_{\midord-1} \in \sadelta_{\midord-1}$,
\ldots,
$\sastate_{\midord'+2}
 \satran{\sastate_{\midord'+1}}
 \sastateset_{\midord'+2} \in \sadelta_{\midord'+2}$
and
$\sastate_{\midord'+1}
 \satran{\sastate'}
 \sastateset_{\midord'+1} \in \sadelta_{\midord'+1}$.

We lift the short-form transition notation to transitions from sets of states.
We assume that state-sets
$\sastates_\maxord, \ldots, \sastates_1$
are disjoint.  Suppose
$\sastateset = \set{\sastate_1,\ldots,\sastate_\numof}$
and for all
$1 \leq \idxi \leq \numof$
we have
$\satranfull{\sastate_\idxi}
            {\cha}
            {\sastateset^\idxi_\branch}
            {\sastateset^\idxi_1,\ldots,\sastateset^\idxi_\midord}$.
Then we have
$\satranfull{\sastateset}
            {\cha}
            {\sastateset_\branch}
            {\sastateset_1,\ldots,\sastateset_\midord}$
where
$\sastateset_\branch
    = \bigcup_{1 \leq \idxi \leq \numof} \sastateset^\idxi_\branch$
and for all $\midord$,
$\sastateset_\midord
    = \bigcup_{1 \leq \idxi \leq \numof} \sastateset^\idxi_\midord$.
Because an annotation can only be of one order, we insist that
$\sastateset_\branch \subseteq \sastates_\midord$
for some $\midord$.

We generalise this to trees as follows.  We write
\[
    \tatranfull{\tastate}
               {\idxi}
               {\numof}
               {\tastate'}
               {\cha}
               {\sastateset_\branch}
               {\sastateset_1,
                \ldots,
                \sastateset_\maxord}
    \quad
    \text{ and }
    \quad
    \tatranfullk{\tastate}
                {\idxi}
                {\numof}
                {\tastate'}
                {\sastate'}
                {\sastateset_{\midord+1},
                 \ldots,
                 \sastateset_\maxord}
\]
when
$\tatran{\tastate}{\idxi}{\numof}{\tastate'}{\sastate}$
and
$\satranfull{\sastate}
            {\cha}
            {\sastateset_\branch}
            {\sastateset_1,
             \ldots
             \sastateset_\maxord}$
or, respectively,
$\satranfullk{\sastate}
             {\sastate'}
             {\sastateset_{\midord+1},
              \ldots
              \sastateset_\maxord}$.

Finally, we remark that a transition to the empty set is distinct
from having no transition.

\section{Backwards Reachability Analysis}

Fix a GASTRS $\gstrs$ and automaton $\ta_0$ for the remainder of the article.
We define
\[
    \prestar{\gstrs}{\ta_0} = \setcomp{\tree}
                                    {\tree \reaches \tree'
                                     \land
                                     \tree' \in \langof{\ta_0}} \ .
\]
We give a saturation algorithm for computing an automaton $\ta$ such that
$\langof{\ta} = \prestar{\gstrs}{\ta_0}$.  Indeed, we prove the following
theorem.  The upper bound is discussed in the sequel.  The lower bound comes
from alternating higher-order pushdown automata~\cite{CW07} and appears in
Appendix~\ref{sec:lower-bound}.

\begin{theorem}
    Given an order-$\maxord$ GASTRS $\gstrs$ and stack tree automaton $\ta_0$,
    $\prestar{\gstrs}{\ta_0}$
    is regular and computable in $\maxord$-EXPTIME, which is optimal.
\end{theorem}

For technical reasons assume for each $\control$ there
is at most one rule
$\stpop{\control_1, \ldots, \control_\numof}{\control}$.  E.g., we
cannot have
$\stpop{\control_1, \control_2}{\control}$
and
$\stpop{\control'_1, \control'_2}{\control}$.
This is not a real restriction since we can introduce intermediate control
states.  E.g.
$\stpop{\control_1, \control_2}{\control_{1, 2}}$
and
$\gtrule{\control_{1, 2}}{\srew{\cha}{\cha}}{\control}$
and
$\stpop{\control'_1, \control'_2}{\control'_{1, 2}}$
and
$\gtrule{\control'_{1, 2}}{\srew{\cha}{\cha}}{\control}$
for all $\cha \in \salphabet$.

\subsubsection*{Initial States}

We say that all states in $\controls$ are \emph{initial}.  Furthermore, a state
$\sastate$ is initial if there is a transition
$\tatran{\tastate}{\idxi}{\numof}{\tastate'}{\sastate}$
or if there exists a transition
$\sastate' \satran{\sastate} \sastateset$
in some
$\sadelta_\midord$.
We make the assumption that all initial states do not have any incoming
transitions and that they are not final\footnote{Hence automata cannot accept
empty stacks from initial states.  This can be overcome by introducing a
bottom-of-stack symbol.}.  Furthermore, we assume any initial state only appears
on one transition.

\subsubsection*{New Transitions}

When we add a transition
$\tatranfull{\tastate}
            {\idxi}
            {\numof}
            {\tastate'}
            {\cha}
            {\sastateset_\branch}
            {\sastateset_1, \ldots, \sastateset_\maxord}$
to the automaton, then, we add
$\tatran{\tastate}{\idxi}{\numof}{\tastate'}{\sastate_\maxord}$
to $\tadelta$ if it does not exist, else we use the existing $\sastate_\maxord$,
and then for each
$\maxord \geq \midord > 1$,
we add
$\sastate_\midord \satran{\sastate_{\midord-1}} \sastateset_\midord$
to
$\sadelta_\midord$
if a transition between $\sastate_\midord$ and $\sastateset_\midord$ does not
already exist, otherwise we use the existing transition and state
$\sastate_{\midord-1}$; finally, we add
$\sastate_1 \satrancol{\cha}{\sastateset_\branch} \sastateset_1$
to $\sadelta_1$.

\subsubsection*{The Algorithm}

We give the algorithm formally here, with intuitive explanations given in the
follow section.  Saturation is a fixed point algorithm.  We begin with a GASTRS
$\gstrs = \tup{\salphabet, \rules}$
and target set of trees by $\ta_0$.  Then, we apply the saturation function
$\satfn$ and obtain a sequence of automata
$\ta_{\idxi+1} = \ap{\satfn}{\ta_\idxi}$.
The algorithm terminates when
$\ta_{\idxi+1} = \ta_{\idxi}$
in which case we will have
$\langof{\ta_{\idxi+1}} = \prestar{\gstrs}{\ta_0}$.

Following the conventions described above for adding transitions to the
automaton, we can only add a finite number of states to the automaton, which
implies that only a finite number of transitions can be added.  Hence, we must
necessarily reach a fixed point for some $\idxi$.

Given
$\ta_\idxi$,
we define
$\ta_{\idxi+1} = \ap{\satfn}{\ta_\idxi}$
to be the automaton obtained by adding to $\ta_\idxi$ the following transitions
and states.
\begin{itemize}
    \item
        For each rule
        $\stsop{\control}{\srew{\cha}{\chb}}{\control'} \in \rules$
        and transition
        $\tatranfull{\tastate}
                    {\idxj}
                    {\numof}
                    {\control'}
                    {\chb}
                    {\sastateset_\branch}
                    {\sastateset_1, \ldots, \sastateset_\maxord}$
        in $\ta_\idxi$, add to
        $\ta_{\idxi+1}$
        the transition
        $\tatranfull{\tastate}
                    {\idxj}
                    {\numof}
                    {\control}
                    {\cha}
                    {\sastateset_\branch}
                    {\sastateset_1, \ldots, \sastateset_\maxord}$.

    \item
        For each rule
        $\stsop{\control}{\scpush{\midord}}{\control'} \in \rules$,
        transition
        $\tatranfull{\tastate}
                    {\idxj}
                    {\numof}
                    {\control'}
                    {\cha}
                    {\sastateset_\branch}
                    {\sastateset_1, \ldots, \sastateset_\maxord}$,
        and
        $\sastateset_1 \satrancol{\cha}{\sastateset'_\branch} \sastateset'_1$
        in $\ta_\idxi$, add
        \[
            \tatranfull{\tastate}
                       {\idxj}
                       {\numof}
                       {\control}
                       {\cha}
                       {\sastateset'_\branch}
                       {\sastateset'_1,
                        \sastateset_2,
                        \ldots,
                        \sastateset_{\midord-1},
                        \sastateset_\midord \cup \sastateset_\branch,
                        \sastateset_{\midord+1},
                        \ldots,
                        \sastateset_\maxord}
        \]
        to $\ta_{\idxi+1}$ when $\midord > 1$, and
        $\tatranfull{\tastate}
                    {\idxj}
                    {\numof}
                    {\control}
                    {\cha}
                    {\sastateset'_\branch}
                    {\sastateset'_1 \cup \sastateset_\branch,
                     \sastateset_2,
                     \ldots,
                     \sastateset_\maxord}$
        when $\midord = 1$.

    \item
        For each rule
        $\stsop{\control}{\spush{\midord}}{\control'} \in \rules$
        and
        $\tatranfull{\tastate}
                    {\idxj}
                    {\numof}
                    {\control'}
                    {\cha}
                    {\sastateset_\branch}
                    {\sastateset_1, \ldots, \sastateset_\maxord}$
        and
        $\satranfull{\sastateset_\midord}
                    {\cha}
                    {\sastateset'_\branch}
                    {\sastateset'_1,
                     \ldots,
                     \sastateset'_\midord}$
        in $\ta_\idxi$, add to
        $\ta_{\idxi+1}$
        \[
            \tatranfull{\tastate}
                       {\idxj}
                       {\numof}
                       {\control}
                       {\cha}
                       {\sastateset_\branch \cup \sastateset'_\branch}
                       {\sastateset_1 \cup \sastateset'_1,
                        \ldots,
                        \sastateset_{\midord-1} \cup \sastateset'_{\midord-1},
                        \sastateset'_\midord,
                        \sastateset_{\midord+1},
                        \ldots,
                        \sastateset_\maxord} \ .
        \]

    \item
        For each rule
        $\stsop{\control}{\spop{\midord}}{\control'} \in \rules$
        and
        $\tatranfullk{\tastate}
                     {\idxj}
                     {\numof}
                     {\control'}
                     {\sastate_\midord}
                     {\sastateset_{\midord+1},
                      \ldots,
                      \sastateset_\maxord}$
        in $\ta_\idxi$, add to
        $\ta_{\idxi+1}$
        for each
        $\cha \in \salphabet$
        \[
            \tatranfull{\tastate}
                       {\idxj}
                       {\numof}
                       {\control}
                       {\cha}
                       {\emptyset}
                       {\emptyset,
                        \ldots,
                        \emptyset,
                        \set{\sastate_\midord},
                        \sastateset_{\midord+1},
                        \ldots,
                        \sastateset_\maxord} \ .
        \]

    \item
        For each rule
        $\stsop{\control}{\scollapse{\midord}}{\control'} \in \rules$
        and
        $\tatranfullk{\tastate}
                     {\idxj}
                     {\numof}
                     {\control'}
                     {\sastate_\midord}
                     {\sastateset_{\midord+1},
                      \ldots,
                      \sastateset_\maxord}$
        in $\ta_\idxi$, add to
        $\ta_{\idxi+1}$
        for each
        $\cha \in \salphabet$
        \[
            \tatranfull{\tastate}
                       {\idxj}
                       {\numof}
                       {\control}
                       {\cha}
                       {\set{\sastate_\midord}}
                       {\emptyset,
                        \ldots,
                        \emptyset,
                        \sastateset_{\midord+1},
                        \ldots,
                        \sastateset_\maxord} \ .
        \]

    \item
        For each rule
        $\stpush{\control}{\control_1, \ldots, \control_\numof} \in \rules$
        and
        $\tatranfull{\tastate}
                    {\idxj}
                    {\numof'}
                    {\tastate'}
                    {\cha}
                    {\sastateset_\branch}
                    {\sastateset_1, \ldots, \sastateset_\maxord}$
        and
        \[
            \tatranfull{\tastate'}
                       {1}
                       {\numof}
                       {\control_1}
                       {\cha}
                       {\sastateset^1_\branch}
                       {\sastateset^1_1, \ldots, \sastateset^1_\maxord},
            \ldots,
            \tatranfull{\tastate'}
                       {\numof}
                       {\numof}
                       {\control_2}
                       {\cha}
                       {\sastateset^2_\branch}
                       {\sastateset^2_1, \ldots, \sastateset^2_\maxord}
        \]
        in $\ta_\idxi$, add to
        $\ta_{\idxi+1}$
        \[
            \tatranfull{\tastate}
                       {\idxj}
                       {\numof'}
                       {\control}
                       {\cha}
                       {\sastateset'_\branch}
                       {\sastateset'_1, \ldots, \sastateset'_\maxord}
        \]
        where
        $\sastateset'_\branch =
         \sastateset_\branch
         \cup
         \sastateset^1_\branch
         \cup
         \cdots
         \cup
         \sastateset^\numof_\branch$
        and for all $\midord$, we have
        $\sastateset'_\midord =
         \sastateset_1
         \cup
         \sastateset^1_\midord
         \cup
         \cdots
         \cup
         \sastateset^\numof_\midord$.

    \item
        For each rule
        $\stpop{\control_1, \ldots, \control_\numof}{\control} \in \rules$
        and
        $\cha_1, \ldots, \cha_\numof \in \salphabet$
        add to $\ta_{\idxi+1}$ the transitions
        $\tatranfull{\control}
                    {\idxj}
                    {\numof}
                    {\control_\idxj}
                    {\cha_\idxj}
                    {\emptyset}
                    {\emptyset, \ldots, \emptyset}$
        for each $1 \leq \idxj \leq \numof$.
\end{itemize}

\subsubsection{Intuition of the Algorithm}

Since rules may only be applied to the leaves of the tree, the algorithm works
by introducing new initial transitions that are derived from existing initial
transitions.  Consider a tree $\tree$ with a leaf node $\tnode$ labelled by
$\brac{\annot{\chb}{\stack_\branch} \scomp{1} \stack}$.
Suppose this tree were already accepted by the automaton, and the
initial transition
$\tatranfull{\tastate}
            {\idxi}
            {\numof}
            {\control}
            {\chb}
            {\sastateset_\branch}
            {\sastateset_1, \ldots, \sastateset_\maxord}$
is applied to $\tnode$.

If we had a rule
$\stsop{\control'}{\srew{\cha}{\chb}}{\control}$
then we could apply this rule to a tree $\tree'$ that is identical to $\tree$
except $\tnode$ is labelled by
$\brac{\annot{\cha}{\stack_\branch} \scomp{1} \stack}$.
After the application, we would obtain $\tree$.  Thus, if $\tree$ is accepted by
the automaton, then $\tree'$ should be accepted.

The saturation algorithm will derive from the above rule and transition a new
transition
$\tatranfull{\tastate}
            {\idxi}
            {\numof}
            {\control'}
            {\chb}
            {\sastateset_\branch}
            {\sastateset_1, \ldots, \sastateset_\maxord}$.
This transition simply changes the control state and top character of the stack.
Thus, we can substitute this transition into the accepting run of $\tree$ to
build an accepting run of $\tree'$.

For a rule
$\stpop{\control_1}{\control}$
we would introduce a transition
$\tatranfull{\control}
            {1}
            {1}
            {\chb}
            {\control_1}
            {\emptyset}
            {\emptyset, \ldots, \emptyset}$.
We can add this transition to any accepting run of a tree with a leaf with
control state $\control$ and it will have the effect of adding a new node with
control state $\control_1$.  Since we can obtain the original tree by applying
the rule, the extended tree should also be accepted.  The intuition is similar
for the
$\spop{\midord}$
and
$\scollapse{\midord}$
operations.

To understand the intuition for the
$\spush{\midord}$,
$\scpush{\midord}$
and
$\stpush{\control}{\control_1, \ldots, \control_\numof}$
rules, one must observe that these rules, applied backwards, have the effect of
replacing multiple copies of identical stacks with a single stack.  Thus, the
new transitions accept the intersection of the stacks that could have been
accepted by multiple previous transitions: taking the union of two sets of
automaton states means that the intersection of the language must be accepted.

\subsubsection*{Correctness}

We have the following
\easyicalp{%
    property.
}{%
    property, proved in Appendix~\ref{sec:correctness}.
}

\begin{namedproperty}{prop:sat-correct}{Correctness of Saturation}
    Given an order-$\maxord$ GASTRS, saturation runs in $\maxord$-EXPTIME and
    builds an automaton $\ta$ such that
    $\langof{\ta} = \prestar{\gstrs}{\ta_0}$.
\end{namedproperty}
\easyicalp{
    \begin{proof}
        The proof of completeness is given in Lemma~\ref{lem:completeness} and
        soundness is given in Lemma~\ref{lem:soundness}.

        The complexity is derived as follows.  We add at most one transition of the
        form
        $\tatran{\tastate}{\idxi}{\numof}{\control}{\sastate}$
        for each $\tastate$, $\idxi$, $\numof$ and $\control$.  Hence we add at most
        a polynomial number of transitions to $\tadelta$.

        Thus, to $\sadelta_\maxord$ we have a polynomial number of states.  We add at
        most one transition of the form
        $\sastate \satran{\sastate'} \sastateset$
        for each $\sastate$ and set of states $\sastateset$.  Thus we have at most
        an exponential number of transitions in $\sadelta_\maxord$.

        Thus, in $\sastates_\midord$ we have a number of states bounded by a tower
        of exponentials of height
        $(\maxord - \midord)$.
        Since we add at most one transition of the form
        $\sastate \satran{\sastate'} \sastateset$
        for each $\sastate$ and $\sastateset$ we have a number of transitions
        bounded by a tower of exponentials of height
        $(\maxord - \midord + 1)$
        giving the number of states in
        $\sastates_{\midord-1}$.

        Thus, at order-$1$ the number of new transitions is bounded by a tower of
        height $\maxord$, giving the $\maxord$-EXPTIME complexity.
    \end{proof}
}{
    Completeness is easily proved by induction over the length of the run to a
    target configuration.  Soundness generalises the notion, used in ICALP 2012,
    of a ``sound'' stack automaton to trees and requires some non-trivial
    definitions to handle the tree structures.  Finally, the complexity follows
    from the fact that only a polynomial number of states can be added at
    order-$\maxord$, which, due to alternation, blows up by one exponential for
    each level of nesting.
}

\section{Context Bounding}
\label{sec:context-bounding}

\easyicalp{
    In the model discussed so far, communication between different nodes of the tree
    had to be done locally (i.e. from parent to child, via the destruction of
    nodes).  We show that the saturation algorithm can be extended to allow a
    bounded amount of communication between distant nodes of the tree without
    destroying the nodes.

    We begin by defining an extension of our model with global state.  We then show
    that being able to compute $\prestar{\gstrs}{\ta_0}$ can easily be adapted to
    allow a bounded number of global state changes.

    \subsection{GASTRS with Global State}

    \begin{definition}
          [Order-$\maxord$ Ground Annotatee Stack Tree Rewrite Systems with Global State]
        An \emph{order-$\maxord$ ground annotated stack tree rewrite system (GASTRS)
        with global state} $\gstrs$ is a tuple
        $\tup{\salphabet, \controls, \globals, \rules}$
        where $\salphabet$ is a finite stack alphabet, $\controls$ is a finite set
        of control states, $\globals$ is a finite set of global states, and
        $\rules
         \subset
         \globals
         \times
         \stacktreeops{\maxord}{\salphabet}{\controls} \times \globals$
        is a finite set of operations.
    \end{definition}

    A configuration of an order-$\maxord$ GASTRS with global state is a pair
    $\config{\gstate}{\tree}$
    where
    $\gstate \in \globals$
    and $\tree$ is an order-$\maxord$ annotated stack tree over alphabet
    $\salphabet$.  We have a transition
    $\config{\gstate}{\tree} \tran \config{\gstate'}{\tree'}$
    whenever there is some
    $\tup{\gstate, \strule, \gstate'} \in \rules$
    and
    $\tree' \in \ap{\strule}{\tree}$.
    We write
    $\tree \reaches \tree'$ when there is a run
    $\tree
     =
     \tree_0 \tran \cdots \tran \tree_\numof
     =
     \tree'$.
}{}

\subsection{The Context-Bounded Reachability Problem}

The context-bounded reachability problem is to compute the set of configurations
from which there is a run to some target set of configurations, and moreover,
the global state is only changed at most $\lifespan$ times, where $\lifespan$ is
some bound given as part of the input.

\begin{definition}[Global Context-Bounded Backwards Reachability Problem]
    Given a GASTRS with global state $\gstrs$, and a stack tree automaton
    $\ta^0_\gstate$ for each
    $\gstate \in \globals$,
    and a bound $\lifespan$, the \emph{global context-bounded backwards
    reachability problem} is to compute a stack tree automaton
    $\ta_\gstate$ for each
    $\gstate \in \globals$,
    such that
    $\tree \in \langof{\ta_\gstate}$
    iff
    there is a run
    \[
        \config{\gstate}{\tree}
        =
        \config{\gstate_0}{\tree_0}
        \tran \cdots \tran
        \config{\gstate_\numof}{\tree_\numof}
        =
        \config{\gstate'}{\tree'}
    \]
    with
    $\tree' \in \langof{\ta^0_{\gstate'}}$
    and there are at most $\lifespan$ transitions during the run such that
    $\gstate_\idxi \neq \gstate_{\idxi+1}$.
\end{definition}

\subsection{Decidability of Context-Bounded Reachability}

Since the number of global state changes is bounded, the sequence of global
state changes for any run witnessing context-bounded reachability is of the form
$\gstate_0, \ldots, \gstate_\numof$
where
$\numof \leq \lifespan$.
Let $\gstateseqs$ be the set of such sequences.

Suppose we could compute for each such sequence
$\gstateseq = \gstate_0, \ldots, \gstate_\numof$
an automaton
$\ta_\gstateseq$
such that
$\tree \in \langof{\ta_\gstateseq}$
iff there is a run from
$\config{\gstate_0}{\tree}$
to
$\config{\gstate_\numof}{\tree'}$
with
$\tree' \in \langof{\ta_{\gstate_\numof}}$
where the sequence of global states appearing on the run is $\gstateseq$.  We
could then compute an answer to the global context-bounded backwards
reachability problem by taking
\[
    \ta_\gstate
    =
    \bigcup\limits_{\gstate\gstateseq \in \gstateseqs}
        \ta_{\gstate\gstateseq} \ .
\]

To compute $\ta_\gstateseq$ we first make the simplifying assumption (without
loss of generality) that for each
$\gstate \neq \gstate'$
there is a unique
$\tup{\gstate, \strule, \gstate'} \in \rules$
and moreover
$\strule = \stsop{\control}{\srew{\cha}{\chb}}{\control'}$.
Furthermore, for all
$\gstate \in \globals$
we define
$\gstrs_\gstate = \tup{\salphabet, \controls, \rules_\gstate}$
where
\[
    \rules_\gstate = \setcomp{\strule}
                             {\tup{\gstate, \strule, \gstate} \in \rules} \ .
\]

We compute $\ta_\gstateseq$ by backwards induction.  Initially, when
$\gstateseq = \gstate$
we compute
\[
    \ta_\gstateseq = \prestar{\gstrs_\gstate}{\ta_\gstate} \ .
\]
It is immediate to see that
$\ta_\gstateseq$
is correct.  Now, assume we have
$\gstateseq = \gstate\gstateseq'$
and we have already computed
$\ta_{\gstateseq'}$,
we show how to compute
$\ta_\gstateseq$.

The first step is to compute
$\ta'_\gstateseq$
such that
$\tree \in \langof{\ta'_\gstateseq}$
iff
$\config{\gstate}{\tree} \tran \config{\gstate'}{\tree'}$
where
$\gstate'$ is the first state of $\gstateseq'$ and
$\tree' \in \langof{\ta_{\gstateseq'}}$.
That is,
$\ta'_\gstateseq$
accepts all trees from which we can change the current global state to
$\gstate'$.  That is, by a single application of the unique rule
$\tup{\gstate, \strule, \gstate'}$.
Once we have computed this automaton we need simply build
\[
    \ta_\gstateseq = \prestar{\gstrs_\gstate}{\ta'_\gstateseq}
\]
and we are done.

We first define
$\ta''_\gstateseq$
which is a version of
$\ta_{\gstateseq'}$
that has been prepared for a single application of
$\tup{\gstate, \strule, \gstate'}$.
From this we compute
$\ta_\gstateseq$.

The strategy for building
$\ta''_\gstateseq$
is to mark in the states which child, if any, of the node has the global state
change rule applied to its subtree.  At each level of the tree, this marking
information enforces that only one subtree contains the application.  Thus, when
the root is reached, we know there is only one application in the whole tree.
Note, this automaton does not contain any transitions corresponding to
the actual application of the global change rule.  This is added afterwards to
compute
$\ta_\gstateseq$.
Thus, if
\[
    \ta_\gstateseq
    =
    \tup{
        \tastates, \sastates_\maxord, \ldots, \sastates_1,
        \salphabet,
        \tadelta, \sadelta_\maxord, \ldots, \sadelta_1,
        \controls,
        \tafinals', \safinals_\maxord,\ldots,\safinals_1
    }
\]
then
\[
    \ta''_\gstateseq
    =
    \tup{
        \tastates', \sastates_\maxord, \ldots, \sastates_1,
        \salphabet,
        \tadelta', \sadelta_\maxord, \ldots, \sadelta_1,
        \controls,
        \tafinals', \safinals_\maxord,\ldots,\safinals_1
    }
\]
where, letting $\numof$ be the maximum number of children permitted by any
transition of $\ta_\gstateseq$,
\[
    \tastates' = \controls \cup \tastates \times \set{0, \ldots, \numof}
    \quad
    \text{ and }
    \quad
    \tafinals' = \setcomp{\tup{\tastatef, \idxi}}
                         {\tastatef \in \tafinals \land 0 < \idxi \leq \numof}
\]
and we define
\[
    \begin{array}{rcl}
        \tadelta'
        &=&
        \tadeltainit \cup \tadeltanoapp \cup \tadeltapass
        \\
        \\
        \tadeltainit
        &=&
        \setcomp{\tatran{\tup{\tastate, 0}}{\idxi}{\numof}{\control}{\sastate}}
                {\tatran{\tastate}{\idxi}{\numof}{\control}{\sastate} \in \tadelta}
        \cup
        \\
        & &
        \setcomp{\tatran{\tup{\tastate, \idxj}}{\idxi}{\numof}{\control}{\sastate}}
                {\tatran{\tastate}{\idxi}{\numof}{\control}{\sastate} \in \tadelta
                 \land
                 \idxi \neq \idxj}
        \\
        \\
        \tadeltanoapp
        &=&
        \setcomp{\tatran{\tup{\tastate, 0}}
                        {\idxi}
                        {\numof}
                        {\tup{\tastate', 0}}
                        {\sastate}}
                {\tatran{\tastate}{\idxi}{\numof}{\tastate'}{\sastate} \in \tadelta}
        \\
        \\
        \tadeltapass
        &=&
        \setcomp{\tatran{\tup{\tastate, \idxi}}
                        {\idxi}
                        {\numof}
                        {\tup{\tastate, \idxj}}
                        {\sastate}}
                {\tatran{\tastate}{\idxi}{\numof}{\tastate'}{\sastate} \in \tadelta}
        \cup \\
        & &
        \setcomp{\tatran{\tup{\tastate, \idxj}}
                        {\idxi}
                        {\numof}
                        {\tup{\tastate, 0}}
                        {\sastate}}
                {\tatran{\tastate}{\idxi}{\numof}{\tastate'}{\sastate} \in \tadelta
                 \land
                 \idxi \neq \idxj} \ .
    \end{array}
\]
In the above $\tadeltainit$ has two kinds of transitions.  The first set are the
initial transitions for the nodes to which the rewrite rule is not applied
(indicated by the $0$).  The second set are the rules where the rewrite rule is
applied at the $\idxj$th sibling of the $\idxi$th child.  Next $\tadeltanoapp$ are
the transitions for subtrees which have not been marked as containing the
application.  Finally, $\tadeltapass$ propagates information about where the
application actually occurred up the tree.  The first set of transitions in
$\tadeltapass$ are used when the $\idxi$th child contains the application (hence
it labels the parent with the information that the $\idxi$th child contains the
application).  The second set of transitions guess that the $\idxj$th sibling
contains the application.  Thus, at any node, at most one child subtree may
contain the application.  The set of final states enforce that the application
has occurred in some child.

To compute
$\ta'_\gstateseq$,
letting
$\strule = \stsop{\control}{\srew{\cha}{\chb}}{\control'}$
be the operation on the global state change, we add to
$\ta''_\gstateseq$
a transition
\[
    \tatranfull{\tup{\tastate, \idxi}}
               {\idxi}
               {\numof}
               {\control}
               {\cha}
               {\sastateset_\branch}
               {\sastateset_1, \ldots, \sastateset_\maxord}
\]
for each
\[
    \tatranfull{\tastate}
               {\idxi}
               {\numof}
               {\control'}
               {\chb}
               {\sastateset_\branch}
               {\sastateset_1, \ldots, \sastateset_\maxord}
\]
in $\ta_{\gstateseq'}$.

We remark that, as defined,
$\ta_\gstateseq$
does not satisfy the prerequisites
of the saturation algorithm, since initial states reading stacks might have
incoming transitions, and, moreover, an initial state may label more than one
transition.  We can convert
$\ta_\gstateseq$
to the correct format using the automata manipulations in
Appendix~\ref{sec:aut-particulars}.

\begin{lemma}
    We have
    $\tree \in \langof{\ta'_\gstateseq}$
    iff
    $\config{\gstate}{\tree} \tran \config{\gstate'}{\tree'}$
    via a single application of the transition
    $\tup{\gstate, \strule, \gstate'}$
    and
    $\tree' \in \langof{\ta_{\gstateseq'}}$.
\end{lemma}
\begin{proof}
    First, assume
    $\tree \in \langof{\ta'_\gstateseq}$.
    We argue that there is exactly one leaf
    $\tleaf{\tree}{\idxi}$
    read by a transition
    $\tatran{\tup{\tastate, \idxi}}
            {\idxi}
            {\numof}
            {\control}
            {\sastate}$
    and all other leaves are read by some
    $\tatran{\tup{\tastate, 0}}
            {\idxi}
            {\numof}
            {\control}
            {\sastate}$
    or
    $\tatran{\tup{\tastate, \idxj}}
            {\idxi}
            {\numof}
            {\control}
            {\sastate}$
    with $\idxj \neq \idxi$.

    If there is no such $\tleaf{\tree}{\idxi}$ then all leaf nodes are read by
    some
    $\tatran{\tup{\tastate, 0}}
            {\idxi}
            {\numof}
            {\control}
            {\sastate}$.
    Thus, all parents of the leaf nodes are labelled by
    $\tup{\tastate, 0}$.
    Thus, take any node $\tnode$ and assume its children are labelled by some
    $\tup{\tastate, 0}$.
    It must be the case that $\tnode$ is also labelled by some
    $\tup{\tastate, 0}$
    since otherwise it is labelled
    $\tup{\tastate, \idxi}$
    and its $\idxi$th child must be labelled by some
    $\tup{\tastate, \idxj}$
    with
    $\idxj > 0$, which is a contradiction.  Hence, the accepting state of the
    run must also be some
    $\tup{\tastatef, 0}$
    which is not possible.

    If there are two or more leaves labelled by some
    $\tup{\tastate, \idxi}$
    with
    $\idxi > 0$
    then
    each ancestor must also be labelled by some
    $\tup{\tastate, \idxi}$
    with
    $\idxi > 0$.
    Take the nearest common ancestor $\tnode$ and suppose it is labelled
    $\tup{\tastate, \idxi}$.
    However, since it has two children labelled with non-zero second components,
    we must have used a transition
    $\tatran{\tup{\tastate, \idxi}}
            {\idxj}
            {\numof}
            {\tup{\tastate', \idxj'}}
            {\sastate}$
    which, by definition, cannot exist.

    Hence, we have only one leaf
    $\tleaf{\tree}{\idxi}$
    where
    \[
        \tatranfull{\tup{\tastate, \idxi}}
                   {\idxi}
                   {\numof}
                   {\control}
                   {\cha}
                   {\sastateset_\branch}
                   {\sastateset_1, \ldots, \sastateset_\maxord}
    \]
    is used.  Obtain $\tree'$ by applying
    $\stsop{\control}{\srew{\cha}{\chb}}{\control'}$
    at this leaf.  We build an accepting run of
    $\ta_{\gstateseq'}$
    by taking the run of
    $\ta'_\gstateseq$
    over $\tree$, projecting out the second component of each label, and
    replacing the transition used at
    $\tleaf{\tree}{\idxi}$
    with
    \[
        \tatranfull{\tastate}
                   {\idxi}
                   {\numof}
                   {\control'}
                   {\chb}
                   {\sastateset_\branch}
                   {\sastateset_1, \ldots, \sastateset_\maxord} \ .
    \]
    Hence, we are done.

    In the other direction take $\tree$ and $\tree'$ obtained by applying
    $\stsop{\control}{\srew{\cha}{\chb}}{\control'}$
    at leaf
    $\tleaf{\tree}{\idxi}$.
    We take the accepting run of
    $\ta_{\gstateseq'}$
    over $\tree'$
    and build an accepting run of
    $\ta'_\gstateseq$
    over $\tree$.  Let
    \[
        \tatranfull{\tastate}
                   {\idxi}
                   {\numof}
                   {\control'}
                   {\chb}
                   {\sastateset_\branch}
                   {\sastateset_1, \ldots, \sastateset_\maxord} \ .
    \]
    be the transition used at
    $\tleaf{\tree}{\idxi}$.  We replace it with
    \[
        \tatranfull{\tup{\tastate, \idxi}}
                   {\idxi}
                   {\numof}
                   {\control}
                   {\cha}
                   {\sastateset_\branch}
                   {\sastateset_1, \ldots, \sastateset_\maxord} \ .
    \]
    Starting from above the root node, let the $\idxj$th child be the first on the
    path to
    $\tleaf{\tree}{\idxi}$
    (the root node is the $1$st child of ``above the root node'').
    For all children except the $\idxj$th, take the
    transition
    $\tatran{\tastate}{\idxj'}{\numof}{\tastate'}{\sastate}$
    used in the run over $\tree'$ and replace it with
    $\tatran{\tup{\tastate, \idxj}}{\idxj'}{\numof}{\tup{\tastate', 0}}{\sastate}$.
    The remainder of the run in the descendents of these children requires us to
    use
    $\tatran{\tup{\tastate, 0}}{\idxi'}{\numof}{\tup{\tastate', 0}}{\sastate}$
    or
    $\tatran{\tup{\tastate, 0}}{\idxi'}{\numof}{\tastate}{\sastate}$
    instead of
    $\tatran{\tastate}{\idxi'}{\numof}{\tastate'}{\sastate}$.

    For the $\idxj$th child, we use instead of
    $\tatran{\tastate}{\idxj}{\numof}{\tastate'}{\sastate}$.
    the transition
    $\tatran{\tup{\tastate, \idxj}}{\idxj}{\numof}{\tup{\tastate', \idxj'}}{\sastate}$
    when the $\idxj'$th child of this child leads to
    $\tleaf{\tree}{\idxi}$
    or
    the previously identified transition when the $\idxj'$th child of this child
    is the leaf.

    We repeat the routine above until we reach
    $\tleaf{\tree}{\idxi}$,
    at which point we've constructed an accepting run of
    $\ta'_\gstateseq$
    over $\tree$.
\end{proof}

By iterating the above procedure, we obtain
\easyicalp{%
    our result.

    \begin{theorem}[Context-Bounded Reachability]
        The global context-bounded backwards reachability problem for GASTRS with
        global state is decidable.
    \end{theorem}
}{%
    \reftheorem{thm:context-bounded}.
}

\section{Conclusions and Future Work}

We gave a saturation algorithm for annotated stack trees -- a generalisation of
annotated pushdown systems with the ability to fork and join threads.  We build
on the saturation method implemented by the \cshore tool.  We would like to
implement this work.  We may also investigate higher-order versions of senescent
ground tree rewrite systems~\cite{H14}, which generalises
scope-bounding~\cite{lTN11} to trees.

    \bibliographystyle{plain}
    \bibliography{../common/references}

    \appendix

\section{Particulars of Annotated Stack Tree Automata}
\label{sec:aut-particulars}

Here we discuss various particulars of our stack tree automata: the definition
of runs, the effective boolean algebra, membership, emptiness, transformations
to normal form, and comparisons with other possible stack tree automata definitions.

\subsection{Definition of Runs over Stacks}

We give a more formal definition of a run accepting a stack.  First we introduce
some notation.

For $\maxord \geq \midord > 1$, we write $\sastateset_1 \satran{\sastateset'}
\sastateset_2$ to denote an order-$\midord$ transition from a set of states
whenever $\sastateset_1 = \set{\sastate_1, \ldots, \sastate_\numof}$ and for
each $1 \leq \idxi \leq \numof$ we have $\sastate_\idxi \satran{\sastate'_\idxi}
\sastateset_\idxi$ and $\sastateset' = \set{\sastate'_1, \ldots,
\sastate'_\numof}$ and $\sastateset_2 = \bigcup_{1 \leq \idxi \leq \numof}
\sastateset_\idxi$.  The analogous notation at order-$1$ is a special case of
the short-form notation defined in Section~\ref{ssec:notations}.

Formally, fix an annotated stack tree automaton
\[
    \ta = \tup{
               \tastates, \sastates_\maxord,\ldots,\sastates_1,
               \salphabet,
               \tadelta, \sadelta_\maxord,\ldots,\sadelta_1,
               \controls,
               \tafinals, \safinals_\maxord,\ldots,\safinals_1
           }
\]
We say a node \emph{contains} a character if its exiting edge is labelled by the
character.  Recall the tree view of an annotated stack, an example of which is
given below.
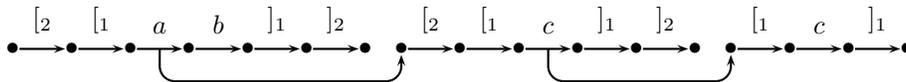
\begin{center}
    \vspace{4ex}
    \begin{psmatrix}[nodealign=true,colsep=2ex,rowsep=2ex]
        \bnode{N1} && \bnode{N2} && \bnode{N3} &\pnode{N34}& \bnode{N4} &&
        \bnode{N5} && \bnode{N6} && \bnode{N7} &

        \bnode{N8} && \bnode{N9} && \bnode{N10} &\pnode{N1011}& \bnode{N11} &&
        \bnode{N12} && \bnode{N13} &

        \bnode{N14} &\pnode{N1415}& \bnode{N15} && \bnode{N16} && \bnode{N17} \\

        \psset{angle=-90,linearc=.2}
        \ncline{->}{N1}{N2}^{$\sopen{2}$}
        \ncline{->}{N2}{N3}^{$\sopen{1}$}
        \ncline{->}{N3}{N4}^{$\cha$}
        \ncbar{->}{N34}{N8}
        \ncline{->}{N4}{N5}^{$\chb$}
        \ncline{->}{N5}{N6}^{$\sclose{1}$}
        \ncline{->}{N6}{N7}^{$\sclose{2}$}

        \ncline{->}{N8}{N9}^{$\sopen{2}$}
        \ncline{->}{N9}{N10}^{$\sopen{1}$}
        \ncline{->}{N10}{N11}^{$\chc$}
        \ncbar{->}{N1011}{N14}
        \ncline{->}{N11}{N12}^{$\sclose{1}$}
        \ncline{->}{N12}{N13}^{$\sclose{2}$}

        \ncline{->}{N14}{N15}^{$\sopen{1}$}
        \ncline{->}{N15}{N16}^{$\chd$}
        \ncline{->}{N16}{N17}^{$\sclose{1}$}
    \end{psmatrix}
\end{center}

Some stack (tree) $\stack$ is accepted by $\ta$ from states
$\sastateset_0 \subseteq \sastates_\midord$
--- written
$\stack \in \slang{\sastateset_0}{\ta}$
--- whenever the nodes of the tree can be labelled by elements of
$\bigcup\limits_{1 \leq \midord' \leq \maxord}
    2^{\sastates_{\midord'}}$
such that
\begin{enumerate}
    \item
        $\sastateset_0$ is a subset of the label of the node containing the
        first
        $\sopen{\midord-1}$
        character of the word, or if
        $\midord = 1$,
        the first character $\cha \in \salphabet$, and

    \item
        for any node containing a character
        $\sopen{\midord'}$
        labelled by $\sastateset$, then for all
        $\sastate_1 \in \sastateset$,
        there exists some transition
        $\tup{\sastate_1, \sastate_2, \sastateset_1} \in \sadelta_{\midord'+1}$
        such that $\sastate_2$ appears in the label of the succeeding node and
        $\sastateset_1$ is a subset of the label of the node succeeding the
        matching
        $\sclose{\midord'}$
        character, and

    \item
        for any node containing a character
        $\sclose{\midord'}$,
        the label $\sastateset$ is a subset of
        $\safinals_{\midord'}$,
        and the final node of an order-$\midord$ stack is labelled by
        $\sastateset \subseteq \safinals_\midord$,
        and

    \item
        for any node containing a character
        $\cha \in \salphabet$,
        labelled by $\sastateset$, for all
        $\sastate' \in \sastateset$,
        there exists some transition
        $\tup{\sastate', \cha, \sastateset_\branch, \sastateset'}
            \in \sadelta_1$ such that $\sastateset_\branch$
        is a subset of the label of the node annotating $\cha$, and
        $\sastateset'$ is a subset of the label of the succeeding node.
\end{enumerate}

That is, a stack automaton is essentially a stack- and annotation-aware
alternating automaton, where annotations are treated as special cases of the
alternation.

\subsection{Effective Boolean Algebra}

In this section we prove the following.

\begin{proposition}
    Annotated stack tree automata form an effective boolean algebra.
\end{proposition}
\begin{proof}
    This follows from Proposition~\ref{prop:aut-union},
    Proposition~\ref{prop:aut-intersect}, and Proposition~\ref{prop:aut-negate}
    below.
\end{proof}

\begin{proposition}
\label{prop:aut-union}
    Given two automata
    \[
        \ta = \tup{
                   \tastates, \sastates_\maxord,\ldots,\sastates_1,
                   \salphabet,
                   \tadelta, \sadelta_\maxord,\ldots,\sadelta_1,
                   \controls,
                   \tafinals, \safinals_\maxord,\ldots,\safinals_1
               }
    \]
    and
    \[
        \ta' = \tup{
                   \tastates', \sastates_\maxord',\ldots,\sastates_1',
                   \salphabet,
                   \tadelta', \sadelta_\maxord',\ldots,\sadelta_1',
                   \controls',
                   \tafinals', \safinals_\maxord',\ldots,\safinals_1'
               }
    \]
    there is an automaton $\ta''$ which recognises the union of the languages
of $\ta$ and $\ta'$.
\end{proposition}
\begin{proof}
 Supposing $\ta$ and $\ta'$ are disjoint except for $\controls$ and no state
 $\control \in \controls$
 has any incoming transition, the automaton we construct is:
 \[
        \ta'' = \tup{\begin{array}{l}
                   \tastates \cup \tastates',
                   \\
                   \sastates_\maxord \cup \sastates_\maxord',
                   \ldots,
                   \sastates_1 \cup \sastates_1',
                   \\
                   \salphabet,
                   \\
                   \tadelta \cup \tadelta',
                   \sadelta_\maxord \cup \sadelta_\maxord',
                   \ldots,
                   \sadelta_1 \cup \sadelta_1',
                   \\
                   \controls,
                   \\
                   \tafinals \cup \tafinals',
                   \safinals_\maxord \cup \safinals_\maxord',
                   \ldots,
                   \safinals_1 \cup \safinals_1'
               \end{array}}
    \]

    Every run in $\ta$ (resp $\ta'$) is a run of $\ta''$ as every state and
transition of $\ta$ is in $\ta''$.

    A run in $\ta''$ is a run of $\ta$ or of $\ta'$, as every state and
transition $\ta''$ is in $\ta$ or in $\ta'$, and as the sets of states and
transitions are disjoint except for initial states (which do not have incoming
transitions), a valid run is either entirely in $\ta$ or in $\ta'$.
\end{proof}

\begin{proposition}
\label{prop:aut-intersect}
    Given two automata
    \[
        \ta = \tup{
                   \tastates, \sastates_\maxord,\ldots,\sastates_1,
                   \salphabet,
                   \tadelta, \sadelta_\maxord,\ldots,\sadelta_1,
                   \controls,
                   \tafinals, \safinals_\maxord,\ldots,\safinals_1
               }
    \]
    and
    \[
        \ta' = \tup{
                   \tastates', \sastates_\maxord',\ldots,\sastates_1',
                   \salphabet,
                   \tadelta', \sadelta_\maxord',\ldots,\sadelta_1',
                   \controls',
                   \tafinals', \safinals_\maxord',\ldots,\safinals_1'
               }
    \]
    there is an automaton $\ta''$ which recognises the intersection of the
languages of $\ta$ and $\ta'$.
\end{proposition}

\begin{proof}
 We construct the following automaton:
    \[
        \ta' = \tup{
                   \tastates'', \sastates_\maxord'',\ldots,\sastates_1'',
                   \salphabet,
                   \tadelta'', \sadelta_\maxord'',\ldots,\sadelta_1'',
                   \controls'',
                   \tafinals'', \safinals_\maxord'',\ldots,\safinals_1''
               }
    \]

    For any pair of states
    $\sastate, \sastate' \in \sastates_\maxord \cup \sastates'_\maxord$
    we can assume a state
    $\sastate \cap \sastate'$
    accepting the intersection of the stacks accepted from $\sastate$ and
    $\sastate'$.  This comes from the fact that stack automata form an effective
    boolean algebra~\cite{BCHS12}.  The states and transitions in
    $\sastates_\maxord'',\ldots,\sastates_1''$,
    $\sadelta_\maxord'',\ldots,\sadelta_1''$,
    and
    $\safinals_\maxord'',\ldots,\safinals_1''$
    come from this construction.

    For $\tastate_1 \in \tastates$ and $\tastate_2 \in \tastates'$, we define
$q_{1,2}$ to be in $\tastates''$ such that, for every
$\tatran{\tastate_1}{\idxi}{\numof}{\tastate_1'}{\sastate_1}$
    and
    $\tatran{\tastate_2}{\idxi}{\numof}{\tastate_2'}{\sastate_2}$
   , we add the transition
   $\tatran{\tastate_{1,2}}{\idxi}{\numof}{\tastate_{1,2}'}{\sastate_1 \cap
\sastate_2}$.

  We have $\tastate_{1,2} \in \tafinals''$ if and only if $\tastate_1 \in \tafinals$
and $\tastate_2 \in \tafinals'$.

  A run exists in $\ta''$ if and only if there is a run in $\ta$ and one in
$\ta'$, by construction.
\end{proof}

\begin{proposition}
\label{prop:aut-negate}
 Given an automaton,
    \[
        \ta = \tup{
                   \tastates, \sastates_\maxord,\ldots,\sastates_1,
                   \salphabet,
                   \tadelta, \sadelta_\maxord,\ldots,\sadelta_1,
                   \controls,
                   \tafinals, \safinals_\maxord,\ldots,\safinals_1
               }
    \]
   there is an automaton $\ta'$ which accepts a tree if and only if it is not
accepted by $\ta$.
\end{proposition}
\begin{proof}
    We define the complement as follows.  We first assume that for each
    $\sastate \in \sastates_\maxord$
    we also have
    $\comp{\sastate} \in \sastates_\maxord$
    that accepts the complement of $\sastate$.  This follows from the
    complementation of stack automata in ICALP 2012~\cite{BCHS12}.

    Then, we define $\ta'$ to be the complement of $\ta$, which contains
    \[
        \ta' = \tup{
                    \tastates', \sastates_\maxord,\ldots,\sastates_1,
                    \salphabet,
                    \tadelta', \sadelta_\maxord,\ldots,\sadelta_1,
                    \controls,
                    \tafinals', \safinals_\maxord,\ldots,\safinals_1
                }
    \]
    where, letting
    $\numof_{\text{max}}$
    be the maximum number of children that can appear in a tree accepted by $\ta$
    (this information is easily obtained from the transitions of $\ta$), we have
    \[
        \tastates' = \bigcup\limits_{\numof \leq \numof_{\text{max}}}
                        \brac{2^{\tastates}}^\numof \ .
    \]
    That is, the automaton will label nodes of the tree with a set of states for
    each child.  The $\idxi$th set will be the set of all labels $\tastate$ that
    could have come from the $\idxi$th child in a run of $\ta$.  Since all children
    have to agree on the $\tastate$ that labels a node, then a label
    $\tup{\tastateset_1, \ldots, \tastateset_\numof}$
    means that the set
    $\tastateset_1 \cap \cdots \cap \tastateset_\numof$
    is the set of states $\tastate$ that could have labelled the node in a run of
    $\ta$.

    The transition relation $\tadelta'$ is the set of transitions of the form
    \[
        \tatran{\tup{\tastateset_1, \ldots, \tastateset_\numof}}
               {\idxi}
               {\numof}
               {\tup{\tastateset'_1, \ldots, \tastateset'_{\numof'}}}
               {\sastate}
    \]
    where
    $\numof, \numof' \leq \numof_{\text{max}}$
    and for all $\idxj \neq \idxi$, the set $\tastateset_\idxj$ is any subset of
    $\tastates$, and
    $\tastateset_\idxi \subseteq \tastates$
    and $\sastate$ are such that
    \begin{itemize}
        \item
            $\sastate = \bigcap\limits_{\tastate \in \tastates} \sastate_\tastate$,
            and

        \item
            if $\tastate \in \tastateset_\idxi$ then
            \[
                \sastate_\tastate = \sastate_1 \cup \cdots \cup \sastate_\numofl
            \]
            where
            $\tatran{\tastate}{\idxi}{\numof}{\tastate_1}{\sastate_1}$,
            \ldots,
            $\tatran{\tastate}{\idxi}{\numof}{\tastate_\numofl}{\sastate_\numofl}$
            are all transitions to $\tastate$ via the $\idxi$th of $\numof$ children
            with the property that
            \[
                \tastate_\idxj \in \tastateset'_1 \cap \cdots \cap \tastateset'_{\numof'}
            \]
            for all $\idxj$.

        \item
            if $\tastate \notin \tastateset_\idxi$ then
            \[
                \sastate_\tastate = \comp{\sastate_1}
                                    \cap \cdots \cap
                                    \comp{\sastate_\numofl}
            \]
            where
            $\tatran{\tastate}{\idxi}{\numof}{\tastate_1}{\sastate_1}$,
            \ldots,
            $\tatran{\tastate}{\idxi}{\numof}{\tastate_\numofl}{\sastate_\numofl}$
            are all transitions to $\tastate$ via the $\idxi$th of $\numof$ children
            with the property that
            \[
                \tastate_\idxj \in \tastateset'_1 \cap \cdots \cap \tastateset'_{\numof'}
            \]
            for all $\idxj$.
    \end{itemize}
    In each transition, the sets $\tastateset_\idxj$ for all $\idxj \neq \idxi$ have
    no constraints.  The automaton effectively guesses the set of labels that could
    have come from sibling nodes.  The set $\tastateset_\idxi$ contains all
    labellings that could have come from the $\idxi$th child given the set of
    labellings that could have labelled the child.  The final condition above
    insists that transitions to any state not in $\tastateset_\idxi$ could not have
    been applied to the child.

    The set of accepting states is
    \[
        \setcomp{\tup{\tastateset_1, \ldots, \tastateset_\numof}}
                {\nexists \tastatef \in \finals .
                    \tastatef \in \tastateset_1 \cap \cdots \cap \tastateset_\numof} \ .
    \]
    For the initial states, we alias $\control = \set{\control}$

    We prove that this automaton is the complement of $\ta$.
    Associate to each node $\tnode$ the
    set $\tastateset_\tnode$ such that
    $\tastate \in \tastateset_\tnode$
    iff there is some (partial, starting from the leaves) run of $\ta$ that
    labels $\tnode$ with $\tastate$.  We prove that all runs of $\ta'$ label
    $\tnode$ with some
    $\tup{\tastateset_1, \ldots, \tastateset_\numof}$
    such that
    $\tastateset_\tnode
        = \tastateset_1 \cap \cdots \cap \tastateset_\numof$.

    At the leaves of the tree this is immediate since $\ta$ must label the
    node with some $\control$, and $\ta'$ must label it with
    $\set{\control}$.

    Now, suppose we have a node $\tnode$ with children $\tnode 1$, \ldots,
    $\tnode \numof$ and the property holds for all children.

    Take some
    $\tastate \in \tastateset_\tnode$.
    Let
    $\tatran{\tastate}{1}{\numof}{\tastate_1}{\sastate_1}$,
    \ldots,
    $\tatran{\tastate}{\numof}{\numof}{\tastate_\numof}{\sastate_\numof}$
    be the transitions used in the run labelling $\tnode$ with $\tastate$.
    For each $\idxi$ we must have by induction $\tastate_\idxi$ appearing in
    all sets labelling $\tnode \idxi$ in a run of $\ta'$.  Now suppose
    $\ta'$ labels $\tnode$ with
    $\tup{\tastateset_1, \ldots, \tastateset_\numof}$
    and moreover
    $\tastate \notin \tastateset_\idxi$.
    Then, by construction, we must have that the stack labelling $\tnode
    \idxi$ is accepted from $\comp{\sastate_\idxi}$.  However, since the
    stack must have been accepted from $\sastate_\idxi$ we have a
    contradiction.  Thus, $\tastate \in \tastateset_\idxi$.

    Now take some
    $\tastate \notin \tastateset_\tnode$.
    Thus, there is some $\idxi$ such that, letting
    $\tatran{\tastate}{\idxi}{\numof}{\tastate_1}{\sastate_1}$,
    \ldots,
    $\tatran{\tastate}{\idxi}{\numof}{\tastate_\numofl}{\sastate_\numofl}$
    be all transitions with $\tastate_\idxj$ appearing in
    $\tastateset_{\tnode\idxi}$,
    we know the stack labelling $\tnode\idxi$ is not accepted from any
    $\sastate_\idxj$ (and is accepted from all
    $\comp{\sastate_\idxj}$).
    Now suppose
    $\ta'$ labels $\tnode$ with
    $\tup{\tastateset_1, \ldots, \tastateset_\numof}$
    and moreover
    $\tastate \in \tastateset_\idxi$.
    Then, by construction, we must have that the stack labelling $\tnode
    \idxi$ is accepted from some $\sastate_\idxj$, which is a contradiction.
    Thus, $\tastate \notin \tastateset_\idxi$.

    Hence
    $\tastateset_\tnode
        = \tastateset_1 \cap \cdots \cap \tastateset_\numof$
    as required.

    Now, assume there is some accepting run of $\ta$ via final state
    $\tastatef$.  Assume there is an accepting run of $\ta'$.  Then necessarily the
    run of $\ta'$ has as its final label some tuple such that
    $\tastatef \in \tastateset_1 \cap \cdots \cap \tastateset_\numof$.
    This contradicts the fact that the run of $\ta'$ is accepting.

    Conversely, take some accepting run of $\ta'$.  The accepting state
    $\tup{\tastateset_1, \ldots, \tastateset_\numof}$
    of this run has no final state
    $\tastatef \in \tastateset_1 \cap \cdots \cap \tastateset_\numof$
    and thus there can be no accepting run of $\ta$.
\end{proof}

\subsection{Membership}

In this section we prove the following.

\begin{proposition}
    The membership problem for annotated stack tree automata is in linear time.
\end{proposition}
\begin{proof}
    We give an algorithm which checks if a tree $\tree$ is recognised by an
    automaton.

    We start by labelling every leaf labelled with control $\control$
    with $\set{\control}$.

    For every node $\tnode$ such that all its sons have been labelled, we label it
    by every state $\tastate$ such that there exist transitions
    $\tatran{\tastate}{1}{\numof}{\tastate_1}{\sastate_1}, \cdots,
    \tatran{\tastate}{\numof}{\numof}{\tastate_\numof}{\sastate_\numof}$ such that
    each son $\tnode \idxi$ is labelled by a set containing $\tastate_\idxi$ and
    the stack labelling $\tnode \idxi$ is accepted by $\sastate_\idxi$.  Note,
    checking the acceptance of a stack from $\sastate_\idxi$ can be done in
    linear time~\cite{BCHS12}.

    If we can label the root by a final state, the tree is accepted (as at each
    step, if we can label a node by a state, there is a run in which it is
    labelled by this state), otherwise, it is not.

    As knowing if a stack is accepted from a given state is linear in the size of
    the stack, and we visit each node once, and explore with it once each possible
    transitions, the complexity of this algorithm is linear in the size of the tree.
\end{proof}

\subsection{Emptiness}

In this section we prove the following.

\begin{proposition}
    The emptiness problem for annotated stack tree automata is in
    PSPACE-complete.
\end{proposition}
\begin{proof}
 We give the following algorithm:

 We set $\Marked = \controls$.

 If there exists a $q$ which is not in $\Marked$ such that, there is some
$\numof$
 such that for each
 $\idxi \leq \numof$
 we have
 $\tatran{\tastate}{\idxi}{\numof}{\tastate'}{\sastate'}$,
 with
 $\tastate' \in \Marked$
 and there exists a stack recognised from $\sastate'$, we add $q$ to $\Marked$.

  We stop when there does not exist such a state.

  If $\Marked \cap \finals = \emptyset$, the recognised language is empty,
otherwise, there is at least one tree recognised.

  There are at most $|\tastates|$ steps in the algorithm, and the complexity of
the emptiness problem for the states $\sastate$ is PSPACE.  Thus, the algorithm
runs in PSPACE.
\end{proof}

\subsection{Automata Transformations}

In this section we show that annotated stack tree automata can always be
transformed to meet the assumptions of the saturation algorithm.

Take a stack tree automaton
\[
    \ta = \tup{
               \tastates, \sastates_\maxord,\ldots,\sastates_1,
               \salphabet,
               \tadelta, \sadelta_\maxord,\ldots,\sadelta_1,
               \controls,
               \tafinals, \safinals_\maxord,\ldots,\safinals_1
           } \ .
\]
We normalise this automaton as follows.  It can be easily seen at each step that
we preserve the language accepted by the automaton.

First we ensure that there are no transitions
\[
    \tatran{\control}{\idxi}{\numof}{\tastate}{\sastate} \ .
\]
We do this by introducing a new state
$\newtastate{\control}$
for each
$\control \in \controls$.
Then, we replace each
\[
    \tatran{\control}{\idxi}{\numof}{\tastate}{\sastate}
\]
with
\[
    \tatran{\newtastate{\control}}{\idxi}{\numof}{\tastate}{\sastate}
\]
and for each
\[
    \tatran{\tastate}{\idxi}{\numof}{\control}{\sastate}
\]
in the resulting automaton, add a transition (not replace)
\[
    \tatran{\tastate}{\idxi}{\numof}{\newtastate{\control}}{\sastate} \ .
\]
Thus, we obtain an automaton with no incoming transitions to any $\control$.

To ensure unique states labelling transitions, we replace each transition
\[
    \tatran{\tastate}{\idxi}{\numof}{\tastate'}{\sastate}
\]
with a transition
\[
    \tatran{\tastate}{\idxi}{\numof}{\tastate'}{\newsastate{\tastate}{\tastate'}}
\]
where there is one
$\newsastate{\tastate}{\tastate'}$
for each pair of states
$\tastate, \tastate'$.
Then when
$\maxord > 1$
we have a transition
$\newsastate{\tastate}{\tastate'} \satran{\sastate'} \sastateset$
for each
$\sastate \satran{\sastate'} \sastateset$.
Notice, if there are multiple possible
$\sastate$
then
$\newsastate{\tastate}{\tastate'} \satran{\sastate'} \sastateset$
accepts the union of their languages.  Furthermore,
$\newsastate{\tastate}{\tastate'}$
has no incoming transitions.  Moreover, we do not remove any transitions from
$\sastate$ but observe that $\sastate$ is no longer initial.
When
$\maxord = 1$
we have a transition
$\newsastate{\sastate}{\sastateset}
 \satrancol{\cha}{\sastateset_\branch}
 \sastateset'$
for each
$\sastate \satrancol{\cha}{\sastateset_\branch} \sastateset'$.

We then iterate from
$\midord = \maxord$
down to
$\midord = 3$
performing a similar transformation to the above.  That is,
we replace each transition in the order-$\midord$ transition set
\[
    \sastate \satran{\sastate'} \sastateset
\]
with a transition
\[
    \sastate \satran{\newsastate{\sastate}{\sastateset}} \sastateset
\]
where there is one
$\newsastate{\sastate}{\sastateset}$
for each pair of $\sastate$  and $\sastateset$.
Then we have a transition
$\newsastate{\sastate}{\sastateset} \satran{\sastate''} \sastateset'$
for each
$\sastate' \satran{\sastate''} \sastateset'$.
Again, if there are multiple possible
$\sastate'$
then
$\newsastate{\sastate}{\sastateset} \satran{\sastate''} \sastateset'$
accepts the union of their languages.  Furthermore,
$\newsastate{\sastate}{\sastateset}$
has no incoming transitions.

Finally, for $\midord = 2$ the procedure is similar.
We replace each transition in the order-$2$ transition set
\[
    \sastate \satran{\sastate'} \sastateset
\]
with a transition
\[
    \sastate \satran{\newsastate{\sastate}{\sastateset}} \sastateset
\]
where there is one
$\newsastate{\sastate}{\sastateset}$
for each pair of $\sastate$  and $\sastateset$.
Then we have a transition
$\newsastate{\sastate}{\sastateset}
 \satrancol{\cha}{\sastateset_\branch}
 \sastateset'$
for each
$\sastate' \satrancol{\cha}{\sastateset_\branch} \sastateset'$.

\subsection{Alternative Tree Automaton Definition}

An alternative definition of stack tree automata would use transitions
\[
    \tastate
    \leftarrow
    \tup{\tastate_1, \sastate_1},
    \ldots,
    \tup{\tastate_\numof, \sastate_\numof}
\]
instead of
\[
    \tatran{\tastate}{1}{\numof}{\tastate_1}{\sastate_1},
    \ldots,
    \tatran{\tastate}{\numof}{\numof}{\tastate_\numof}{\sastate_\numof} \ .
\]
However, due to the dependency such transitions introduce between
$\sastate_1, \ldots, \sastate_\numof$
it is no longer possible to have a unique sequence
$\sastate_1, \ldots, \sastate_\numof$
for each sequence
$\tastate, \tastate_1, \ldots, \tastate_\numof$
(one cannot simply union the candidates for each $\sastate_\idxi$).

For example suppose we had
$\tastate \leftarrow \tup{\tastate_1, \sastate_1}, \tup{\tastate_2, \sastate_2}$
and
$\tastate \leftarrow \tup{\tastate_1, \sastate'_1}, \tup{\tastate_2, \sastate'_2}$
where $\sastate_1$ accepts $\stack_1$, $\sastate'_1$ accepts $\stack'_1$,
$\sastate_2$ accepts $\stack_2$, and $\sastate'_2$ accepts $\stack'_2$.
If we were to replace these two transitions with
$\tastate \leftarrow \tup{\tastate_1, \sastate_1 \cup \sastate'_1},
                     \tup{\tastate_2, \sastate_2 \cup \sastate'_2}$
we would mix up the two transitions, allowing, for example, the first child to
be labelled by $\stack_1$ and the second by $\stack'_2$.

At a first glance, our tree automaton model may appear weaker since we cannot
enforce dependencies between the candidate
$\sastate_\idxi$s in
\[
    \tatran{\tastate}{1}{\numof}{\tastate_1}{\sastate_1},
    \ldots,
    \tatran{\tastate}{\numof}{\numof}{\tastate_\numof}{\sastate_\numof} \ .
\]
However, it turns out that we can overcome this problem with new copies of
$\tastate$.

That is, suppose we had a set $\tadelta$ of transitions of the form
\[
    \tastate
    \leftarrow
    \tup{\tastate_1, \sastate_1},
    \ldots,
    \tup{\tastate_\numof, \sastate_\numof} \ .
\]
We could simulate the resulting tree automaton using our model by introducing
a state
$\tup{\tastate, \tatrant}$
for each $\tastate$ and $\tatrant$.

Given a transition $\tatrant$ of the above form, we can use a family of rules
\[
    \tatran{\tup{\tastate, \tatrant}}
           {1}
           {\numof}
           {\tup{\tastate_1, \tatrant_1}}
           {\sastate_1},
    \ldots,
    \tatran{\tup{\tastate, \tatrant}}
           {\numof}
           {\numof}
           {\tup{\tastate_\numof, \tatrant_\numof}}
           {\sastate_\numof}
\]
for all sequences
$\tatrant_1, \ldots, \tatrant_\numof$
of
$\tadelta$.  (Note that, although there are an exponential number of such
families, we can create them all from a polynomial number of transitions).  Note
that when
$\tastate_\idxi = \control$
we would use $\control$ on the right hand side instead of
$\tup{\tastate_\idxi, \tatrant_\idxi}$
(recalling that $\control$ has no incoming transitions).

\section{Completeness of Saturation}
\label{sec:completeness}

\begin{namedlemma}{lem:completeness}{Completeness of Saturation}
    The automaton $\ta$ obtained by saturation from $\ta_0$ is such that
    $\prestar{\gstrs}{\ta_0} \subseteq \langof{\ta}$.
\end{namedlemma}
\begin{proof}
    Completeness is proved via a straightforward induction over the length of
    the run witnessing
    $\tree \in \prestar{\gstrs}{\ta_0}$.
    In the base case we have
    $\tree \in \langof{\ta_0}$
    and since $\ta$ was obtained only by adding transitions to $\ta_0$, we are
    done.

    For the induction, take
    $\tree \in \ap{\strule}{\tree'}$
    where
    $\tree' \in \prestar{\gstrs}{\ta_0}$
    and by induction $\ta$ has an accepting run of $\tree'$.  We show how the
    transitions added by saturation can be used to build from the run over
    $\tree'$ an accepting run over $\tree$.

    We first consider the cases where $\strule$ adds or removes nodes to/from the
    tree.  The remaining cases when the stack contents are altered are almost
    identical to the ICALP 2012 proof, and hence are left until the end for the
    interested reader.
    \begin{itemize}
        \item
            When
            $\strule = \stpush{\control}{\control_1, \ldots, \control_\numof}$
            was applied to node
            $\tleaf{\tree}{\idxj}$
            of $\tree$, we have
            \[
                \tree' = \treemod{\treemod{\treemod{\tree}
                                                   {\tleaf{\tree}{\idxj}}
                                                   {\stack}}
                                          {\tleaf{\tree}{\idxj}1}
                                          {\tup{\control_1, \stack}}
                                  \cdots}
                         {\tleaf{\tree}{\idxj}\numof}
                         {\tup{\control_\numof, \stack}}
            \]
            where
            $\tup{\control, \stack}$
            labelled
            $\tleaf{\tree}{\idxj}$.

            Take the initial transitions over
            $\tleaf{\tree}{\idxj}$
            and
            $\tleaf{\tree}{\idxj}1$
            to
            $\tleaf{\tree}{\idxj}\numof$
            of the accepting run of $\tree'$
            \[
                \tatranfull{\tastate}
                           {\idxi}
                           {\numof'}
                           {\tastate_1}
                           {\cha}
                           {\sastateset_\branch}
                           {\sastateset_1, \ldots, \sastateset_\maxord}
            \]
            and
            \[
                \tatranfull{\tastate_1}
                           {1}
                           {\numof}
                           {\control_1}
                           {\cha}
                           {\sastateset^1_\branch}
                           {\sastateset^1_1, \ldots, \sastateset^1_\maxord},
                \ldots,
                \tatranfull{\tastate_1}
                           {\numof}
                           {\numof}
                           {\control_\numof}
                           {\cha}
                           {\sastateset^\numof_\branch}
                           {\sastateset^\numof_1, \ldots, \sastateset^\numof_\maxord}
            \]
            where the components of $\stack$ were accepted from
            $\sastateset_\branch$,
            $\sastateset_1, \ldots, \sastateset_\maxord$
            and
            $\sastateset^1_\branch$,
            $\sastateset^1_1, \ldots, \sastateset^1_\maxord$,
            \ldots,
            $\sastateset^\numof_\branch$,
            $\sastateset^\numof_1, \ldots, \sastateset^\numof_\maxord$.

            By saturation we also have
            \[
                \tatranfull{\tastate}
                           {\idxi}
                           {\numof'}
                           {\control}
                           {\cha}
                           {\sastateset'_\branch}
                           {\sastateset'_1, \ldots, \sastateset'_\maxord}
            \]
            where
            $\sastateset'_\branch =
             \sastateset_\branch
             \cup
             \sastateset^1_\branch
             \cup
             \cdots
             \cup
             \sastateset^\numof_\branch$
            and for all $\midord$, we have
            $\sastateset'_\midord =
             \sastateset_1
             \cup
             \sastateset^1_\midord
             \cup
             \cdots
             \cup
             \sastateset^\numof_\midord$
            from which we obtain a run of $\ta$ over $\tree$ by simply replacing
            the transitions of the run over $\tree'$ identified above with
            $\tatrant$.

        \item
            When
            $\strule = \stpop{\control_1, \ldots, \control_\numof}{\control}$
            was applied to nodes
            $\tleaf{\tree}{\idxj}$
            to
            $\tleaf{\tree}{\idxj+\numof-1}$
            of $\tree$, we have
            $\tree' = \treedel{\tree}
                              {\set{\tleaf{\tree}{\idxj},
                                    \ldots,
                                    \tleaf{\tree}{\idxj+\numof-1}}}$
            and
            $\tleaf{\tree}{\idxj}$,
            \ldots,
            $\tleaf{\tree}{\idxj+\numof}$
            were the only children of their parent $\tnode$.  Moreover, let
            $\tup{\control_1, \stack_1}$
            label
            $\tleaf{\tree}{\idxj}$,
            and \ldots and,
            $\tup{\control_\numof, \stack_\numof}$
            label
            $\tleaf{\tree}{\idxj+\numof-1}$
            and $\tnode$ have the stack $\stack$
            in $\tree$
            and
            $\tup{\control, \stack}$
            label $\tnode$ in $\tree'$.

            The initial transition over
            $\tnode$
            of the accepting run of $\tree'$ was from state $\control$
            By saturation we have
            \[
                \tatrant_1
                =
                \tatranfull{\control}
                           {1}
                           {\numof}
                           {\control_1}
                           {\cha_1}
                           {\emptyset}
                           {\emptyset,
                            \ldots,
                            \emptyset},
                \quad
                \ldots,
                \quad
                \tatrant_\numof
                =
                \tatranfull{\control}
                           {\numof}
                           {\numof}
                           {\control_\numof}
                           {\cha_\numof}
                           {\emptyset}
                           {\emptyset,
                            \ldots,
                            \emptyset}
            \]
            for the
            $\cha_1, \ldots, \cha_\numof$
            at the top of
            $\stack_1$, \ldots, $\stack_\numof$
            respectively.  We get from this a run of $\ta$ over $\tree$ by
            adding $\tatrant_1$ to $\tatrant_\numof$ to the run over $\tree'$ to
            read the nodes
            $\tleaf{\tree}{\idxj}$
            to
            $\tleaf{\tree}{\idxj+\numof-1}$.
    \end{itemize}

    We now consider the cases where $\strule$ applies a stack operation to a
    single node
    $\tleaf{\tree'}{\idxj}$
    of $\tree'$.  Let
    \[
        \tatrant' = \tatranfull{\tastate}
                               {\idxi}
                               {\numof}
                               {\control'}
                               {\cha}
                               {\sastateset_\branch}
                               {\sastateset_1, \ldots, \sastateset_\numof}
    \]
    be the transition applied at node
    $\tleaf{\tree'}{\idxj}$
    in the run.  Additionally, let $\stack'$ be the stack labelling the node,
    and $\control'$ be the control state.

    There is a case for each type of stack operation, all of which
    are almost identical to the ICALP 2012 proof.  In all cases below, $\tree$
    has the same tree structure as $\tree'$ and only differs on the labelling of
    $\tleaf{\tree'}{\idxj} = \tleaf{\tree}{\idxj}$.
    \begin{itemize}
        \item
            When
            $\strule = \gtrule{\control}{\srew{\chb}{\cha}}{\control'}$
            then we also added the transition
            \[
                \tatrant = \tatranfull{\tastate}
                                      {\idxi}
                                      {\numof}
                                      {\control}
                                      {\chb}
                                      {\sastateset_\branch}
                                      {\sastateset_1, \ldots, \sastateset_\numof}
            \]
            to $\ta$.  We have
            \[
                \stack' = \annot{\cha}{\stack_\branch}
                          \scomp{1}
                          \stack_1
                          \scomp{2}
                          \cdots
                          \scomp{\maxord}
                          \stack_\maxord
            \]
            and since
            $\tleaf{\tree}{\idxj}$ is labelled by $\control$ and the stack
            \[
                \stack = \annot{\chb}{\stack_\branch}
                         \scomp{1}
                         \stack_1
                         \scomp{2}
                         \cdots
                         \scomp{\maxord}
                         \stack_\maxord
            \]
            we obtain an accepting run of $\tree$ by simply replacing the
            application of $\tatrant'$ with $\tatrant$.

        \item
            When
            $\strule = \gtrule{\control}{\scpush{\midord}}{\control'}$
            then when $\midord > 1$ we have
            \[
                \stack' = \annot{\cha}{\stack_\midord}
                          \scomp{1}
                          \annot{\cha}{\stack_\branch}
                          \stack_1
                          \scomp{2}
                          \cdots
                          \scomp{\maxord}
                          \stack_\maxord \ .
            \]
            Let
            \[
                \satranfull{\sastateset_1}
                           {\cha}
                           {\sastateset'_\branch}
                           {\sastateset'_1}
            \]
            be the first transitions used to accept
            $\annot{\cha}{\stack_\branch}$.
            From the saturation algorithm we also added
            \[
                \tatrant = \tatranfull{\tastate}
                                      {\idxi}
                                      {\numof}
                                      {\control}
                                      {\cha}
                                      {\sastateset'_\branch}
                                      {\sastateset'_1,
                                       \sastateset_2,
                                       \ldots,
                                       \sastateset_{\midord-1},
                                       \sastateset_\midord
                                           \cup \sastateset_\branch,
                                       \sastateset_{\midord+1},
                                       \ldots,
                                       \sastateset_\maxord}
            \]
            to $\ta$.  Since
            $\tleaf{\tree}{\idxj}$ is labelled by $\control$ and the stack
            \[
                \stack = \annot{\cha}{\stack_\branch}
                         \scomp{1}
                         \stack_1
                         \scomp{2}
                         \cdots
                         \scomp{\maxord}
                         \stack_\maxord
            \]
            we obtain an accepting run of $\tree$ by replacing the application
            of $\tatrant'$ with $\tatrant$.  This follows because $\stack'_1$
            was accepted from $\sastateset'_1$, $\stack_\branch$ from
            $\sastateset'_\branch$ and
            $\stack_\midord$
            was accepted from both $\sastateset_\midord$ and
            $\sastateset_\branch$.

            When $\midord = 1$ we have
            \[
                \stack' = \annot{\cha}{\stack_1}
                          \scomp{1}
                          \annot{\cha}{\stack_\branch}
                          \stack_1
                          \scomp{2}
                          \cdots
                          \scomp{\maxord}
                          \stack_\maxord \ .
            \]
            Let
            \[
                \satranfull{\sastateset_1}
                           {\cha}
                           {\sastateset'_\branch}
                           {\sastateset'_1}
            \]
            be the first transitions used to accept
            $\annot{\cha}{\stack_\branch}$.
            From the saturation algorithm we also added
            \[
                \tatrant = \tatranfull{\tastate}
                                      {\idxi}
                                      {\numof}
                                      {\control}
                                      {\cha}
                                      {\sastateset'_\branch}
                                      {\sastateset'_1 \cup \sastateset_\branch,
                                       \sastateset_2,
                                       \ldots,
                                       \sastateset_\maxord}
            \]
            to $\ta$.  Since
            $\tleaf{\tree}{\idxj}$ is labelled by $\control$ and the stack
            \[
                \stack = \annot{\cha}{\stack_\branch}
                         \scomp{1}
                         \stack_1
                         \scomp{2}
                         \cdots
                         \scomp{\maxord}
                         \stack_\maxord
            \]
            we obtain an accepting run of $\tree$ by replacing the application
            of $\tatrant'$ with $\tatrant$.  This follows because $\stack'_1$
            was accepted from $\sastateset'_1$, $\stack_\branch$ from
            $\sastateset'_\branch$ and
            $\stack_\midord$
            was accepted from both $\sastateset_\midord$ and
            $\sastateset_\branch$.

        \item
            When
            $\strule = \gtrule{\control}{\spush{\midord}}{\control'}$
            then we have
            \[
                \stack' = \stack_\midord
                          \scomp{\midord}
                          \stack_\midord
                          \scomp{\midord+1}
                          \stack_{\midord+1}
                          \cdots
                          \scomp{\maxord}
                          \stack_\maxord
                \quad
                \text{ and }
                \quad
                \stack_\midord = \annot{\cha}{\stack_\branch}
                                 \scomp{1}
                                 \stack'_1
                                 \scomp{2}
                                 \cdots
                                 \scomp{(\midord-1)}
                                 \stack_{\midord-1} \ .
            \]
            Let
            \[
                \satranfull{\sastateset_\midord}
                           {\cha}
                           {\sastateset'_\branch}
                           {\sastateset'_1, \ldots, \sastateset'_\midord}
            \]
            be the transitions use to accept the first character of the second
            appearance of $\stack_\midord$.  From the saturation algorithm we
            also added $\tatrant =$
            \[
                \tatranfull{\tastate}
                           {\idxi}
                           {\numof}
                           {\control}
                           {\cha}
                           {\sastateset_\branch
                                \cup \sastateset'_\branch}
                           {\sastateset_1 \cup \sastateset'_1,
                            \sastateset_2 \cup \sastateset'_2,
                            \ldots,
                            \sastateset_{\midord-1}
                                \cup \sastateset'_{\midord-1},
                            \sastateset'_\midord,
                            \sastateset_{\midord+1},
                            \ldots,
                            \sastateset_\maxord}
            \]
            to $\ta$.  Since
            $\tleaf{\tree}{\idxj}$ is labelled by $\control$ and the stack
            \[
                \stack = \annot{\cha}{\stack_\branch}
                         \scomp{1}
                         \stack_1
                         \scomp{2}
                         \cdots
                         \scomp{\maxord}
                         \stack_\maxord
            \]
            we obtain an accepting run of $\tree$ by replacing the application
            of $\tatrant'$ with $\tatrant$.  This follows because stacks
            $\stack_1$ to $\stack_{\midord-1}$
            are accepted from $\sastateset_1$ and $\sastateset'_1$ to
            $\sastateset_{\midord-1}$
            and
            $\sastateset'_{\midord-1}$
            respectively, $\stack_\branch$ from $\sastateset_\branch$ and
            $\sastateset'_\branch$, and the remainder of the stack from
            $\sastateset'_\midord$,
            $\sastateset_{\midord+1}$,
            \ldots,
            $\sastateset_\maxord$.

        \item
            When
            $\strule = \gtrule{\control}{\spop{\midord}}{\control'}$
            Then we have
            \[
                \stack' = \stack_\midord
                          \scomp{\midord+1}
                          \stack_{\midord+1}
                          \cdots
                          \scomp{\maxord}
                          \stack_\maxord
            \]
            and
            \[
                \stack = \annot{\cha}{\stack_\branch}
                         \scomp{1}
                         \stack_1
                         \scomp{2}
                         \cdots
                         \scomp{\maxord}
                         \stack_\maxord
            \]
            for some $\cha$, $\stack_\branch$, $\stack_1$, \ldots,
            $\stack_{\midord-1}$.
            We break down $\tatrant'$ to find $\sastate_\midord$ such that
            \[
                \tatranfullk{\tastate}
                            {\idxi}
                            {\numof}
                            {\control'}
                            {\sastate_\midord}
                            {\sastateset_{\midord+1},
                             \ldots
                             \sastateset_\maxord}
            \]
            where $\sastate_\midord$ accepts $\stack_\midord$ and
            $\sastateset_{\midord+1}$
            through to
            $\sastateset_\maxord$
            accept
            $\stack_{\midord+1}$
            through to $\stack_\maxord$ respectively.  By saturation we added
            the transition
            \[
                \tatrant = \tatranfull{\tastate}
                                      {\idxi}
                                      {\numof}
                                      {\control}
                                      {\cha}
                                      {\emptyset}
                                      {\emptyset,
                                       \ldots,
                                       \emptyset,
                                       \set{\sastate_\midord},
                                       \sastateset_{\midord+1},
                                       \ldots,
                                       \sastateset_\maxord}
            \]
            from which we obtain an accepting run of $\stack$ with $\control$ as
            required.

        \item
            When
            $\strule = \gtrule{\control}{\scollapse{\midord}}{\control'}$
            Then we have
            \[
                \stack' = \stack_\branch,
                          \scomp{\midord+1}
                          \stack_{\midord+1}
                          \cdots
                          \scomp{\maxord}
                          \stack_\maxord
            \]
            and
            \[
                \stack = \annot{\cha}{\stack_\branch}
                         \scomp{1}
                         \stack_1
                         \scomp{2}
                         \cdots
                         \scomp{\maxord}
                         \stack_\maxord
            \]
            for some $\cha$, $\stack_\branch$, $\stack_1$, \ldots,
            $\stack_\midord$.
            We break down $\tatrant'$ to find $\sastate_\branch$ such that
            \[
                \tatranfullk{\tastate}
                            {\idxi}
                            {\numof}
                            {\control'}
                            {\sastate_\branch}
                            {\sastateset_{\midord+1},
                             \ldots
                             \sastateset_\maxord}
            \]
            where $\sastate_\midord$ accepts $\stack_\branch$ and
            $\sastateset_{\midord+1}$
            through to
            $\sastateset_\maxord$
            accept
            $\stack_{\midord+1}$
            through to $\stack_\maxord$ respectively.  By saturation we added
            the transition
            \[
                \tatrant = \tatranfull{\tastate}
                                      {\idxi}
                                      {\numof}
                                      {\control}
                                      {\cha}
                                      {\set{\sastate_\branch}}
                                      {\emptyset,
                                       \ldots,
                                       \emptyset,
                                       \sastateset_{\midord+1},
                                       \ldots,
                                       \sastateset_\maxord}
            \]
            from which we obtain an accepting run of $\stack$ with $\control$ as
            required.
    \end{itemize}
    Thus, in all cases we find an accepting run of $\ta$, which completes the
    proof.
\end{proof}

\section{Soundness of Saturation}

We prove that the automaton $\ta$ constructed by saturation only accepts trees
in
$\prestar{\gstrs}{\ta_0}$.
The proof relies on the notion of a ``sound'' automaton.  There are several
stages to the proof.
\begin{itemize}
    \item
        We assign meanings to each state of the automaton that ultimately
        capture inclusion in
        $\prestar{\gstrs}{\ta_0}$.

    \item
        We use these meanings to derive a notion of sound transitions.

    \item
        We define a sound automaton based on the notion of sound transitions.

    \item
        We show sound tree automata only accept trees in
        $\prestar{\gstrs}{\ta_0}$.

    \item
        We show the initial automaton $\ta_0$ is sound, and moreover, each
        saturation step preserves soundness, from which we conclude soundness of
        the saturation algorithm.
\end{itemize}

To define the meanings of the states we need to reason
about partial runs of our stack tree automata.  Hence for a tree automaton $\ta$
we define
\[
    \tweaklang{\ta}
\]
to accept trees over the set of control states $\tastates$ (instead of
$\controls$).  That is, we can accept prefixes of trees accepted by $\ta$ by
labelling the leaves with the states that would have appeared on an accepting
run of the full tree.

Furthermore, we write
\[
    \tlang{\tastate_1, \ldots, \tastate_\numof}{\ta}
\]
to denote the set of trees $\tree$ in
$\tweaklang{\ta}$
such that $\tree$ has $\numof$ leaves and the ``control'' states (which now
includes all states in $\tastates$) appearing on the leaves are
$\tastate_1, \ldots, \tastate_\numof$
respectively.  As a special case,
$\tlang{\tastatef}{\ta}$
for all
$\tastatef \in \tafinals$
contains only the empty tree.

\subsection{Meaning of a State}

We assign to each state of the automaton a ``meaning''.  This meaning captures
the requirement that the states $\control$ of the automaton should accept
$\prestar{\gstrs}{\ta_0}$, while the meanings of the non-initial states are
given by the automaton itself (i.e. the states should accept everything they
accept).  For states accepting stacks, the non-initial states again have the
trivial meaning (they should accept what they accept), while the meanings of the
initial states are inherited from the transitions that they label.

We write $\tastateseq$ to denote a sequence
$\tastate_1, \ldots, \tastate_\numof$
and
$\seqlen{\tastate_1, \ldots, \tastate_\numof}$
is $\numof$.

Let $\tenv$ be a partial mapping of nodes to states in $\tastates$, let
$\tenvempty$ be the empty mapping, and let
\[
    \ap{\envmod{\tenv}{\tnode}{\tastate}}{\tnode'} =
    \begin{cases}
        \tastate & \tnode = \tnode' \\
        \ap{\tenv}{\tnode'} & \tnode \neq \tnode' \ .
    \end{cases}
\]
We use these mappings in definition below to place conditions on nodes in the
tree that restrict runs witnessing membership in
$\prestar{\gstrs}{\ta_0}$.

\begin{definition}[$\tree \tmodels{\tenv} \tastate_1, \ldots, \tastate_\numof$]
    If $\tree$ has $\numof$ leaves labelled
    $\tastate_1, \ldots, \tastate_\numof$ respectively
    then
    $\tree \tmodels{\tenv} \tastate_1, \ldots, \tastate_\numof$
    whenever
    $\tree \in \prestar{\gstrs}{\tweaklang{\ta_0}}$
    and there is a run to some
    $\tree' \in \tweaklang{\ta_0}$
    such that -- fixing an accepting run of $\ta_0$ over $\tree'$ -- for all
    nodes $\tnode$ of $\tree$ with
    $\ap{\tenv}{\tnode} = \tastate$,
    then
    \begin{itemize}
        \item
            if
            $\tastate \in \controls$
            then $\tnode$ appears as a leaf during the run and on the first such
            tree in the run, $\tnode$ has control state $\tastate$.

        \item
            if
            $\tastate \notin \controls$
            then $\tnode$ is not a leaf of any tree on the run and the
            accepting run of $\ta$ over $\tree'$ labels $\tnode$ with
            $\tastate$.
    \end{itemize}
    As a special case, when $\tree$ is empty we have
    $\tree \tmodels{\tenvempty} \tastatef$
    and
    $\tastatef \in \tafinals$.
\end{definition}

Once we have assigned meanings to the states of $\tastates$, we need to derive
meanings for the states in
$\sastates_\maxord, \ldots, \sastates_1$.
We first introduce some notation.
\[
    \tree
    \treeplus{\idxi}
    \tup{\tastate_1, \stack_1}, \ldots, \tup{\tastate_\numof, \stack_\numof}
    =
    \treemod{
        \treemod{
            \treemod{\tree}
                    {\tleaf{\tree}{\idxi}}
                    {\stack}
        }
        {\tleaf{\tree}{\idxi}1}
        {\tup{\tastate_1, \stack_1}}
        \cdots
    }
    {\tleaf{\tree}{\idxi}\numof}
    {\tup{\tastate_\numof, \stack_\numof}}
\]
when $\tree$ is non-empty and $\stack$ is the stack labelling
$\tleaf{\tree}{\idxi}$ in $\tree$.  When $\tree$ is empty we have
\[
    \tree \treeplus{0} \tup{\tastate_1, \stack_1}
\]
is the single-node tree labelled by $\tup{\tastate_1, \stack_1}$.

In the definition below we assign meanings to states accepting stacks.  The
first case is the simple case where a state is non-initial, and its meaning is
to accept the set of stacks it accepts.

The second case derives a meaning of a state in $\sastates_\midord$ by
inheriting the meaning from the states of $\sastates_{\midord+1}$.  Intuitively,
if we have a transition
$\sastate_{\midord+1} \satran{\sastate_\midord} \sastateset_{\midord+1}$
then the meaning of $\sastate_\midord$ is that it should accept all stacks that
could appear on top of a stack in the meaning of $\sastateset_{\midord+1}$ to
form a stack in the meaning of $\sastate_\midord$.

The final case is a generalisation of the above case to trees.  The states in
$\sastates_\maxord$ should accept all stacks that could appear on a node of the
tree consistent with a run of the stack tree automaton and the meanings of the
states in $\tastates$.

\begin{definition}[$\stack \smodels \sastate$]
    For any
    $\sastateset \subseteq \sastates_\midord$
    and any order-$\midord$ stack $\stack$, we write
    $\stack \smodels \sastateset$
    if
    $\stack \smodels \sastate$
    for all
    $\sastate \in \sastateset$.
    We define
    $\stack \smodels \sastate$
    by a case distinction on
    $\sastate$.
    \begin{enumerate}
        \item
            When $\sastate$ is a non-initial state in $\sastates_\midord$, then
            we have
            $\stack \smodels \sastate$
            if $\stack$ is accepted from $\sastate$.

        \item
            If $\sastate_\midord$ is an initial state in $\sastates_\midord$ with
            $\midord < \maxord$
            labelling a transition
            $\sastate_{\midord+1} \satran{\sastate_\midord} \sastateset_{\midord+1}
                \in \sadelta_{\midord+1}$
            then we have
            $\stack \smodels \sastate_\midord$
            if for all stacks $\stack'$ such that
            $\stack' \smodels \sastateset_{\midord+1}$
            we have
            $\stack \scomp{\midord+1} \stack' \smodels \sastate_{\midord+1}$.

        \item
        \label{item:order-n-states}
            We have
            $\stack \smodels \sastate$
            where
            $\tatran{\tastate}{\idxi}{\numof}{\tastate'}{\sastate}$
            if for all transitions
            \[
                \tatran{\tastate}{1}{\numof}{\tastate_1}{\sastate_1},
                \ldots,
                \tatran{\tastate}{\numof}{\numof}{\tastate_\numof}{\sastate_\numof}
            \]
            trees
            $\tree \tmodels{\tenv} \tastateseq_1, \tastate, \tastateseq_2$
            and stacks
            $\stack_1, \ldots, \stack_\numof$
            such that
            \[
                \tree
                \treeplus{\idxj}
                \tup{\tastate_1, \stack_1},
                \ldots,
                \tup{\tastate_\numof, \stack_\numof}
                \tmodelsm{\tenv}{\tleaf{\tree}{\idxj}}{\tastate}
                \tastateseq_1,
                \tastate_1, \ldots, \tastate_\numof,
                \tastateseq_2
            \]
            where
            $\idxj = \seqlen{\tastateseq_1} + 1$,
            we have
            \begin{multline*}
                \tree
                \treeplus{\idxj}
                \tup{\tastate_1, \stack_1},
                \ldots,
                \tup{\tastate_{\idxi-1}, \stack_{\idxi-1}},
                \tup{\tastate', \stack},
                \tup{\tastate_{\idxi+1}, \stack_{\idxi+1}},
                \ldots,
                \tup{\tastate_\numof, \stack_\numof}
                \\
                \tmodelsm{\tenv}{\tleaf{\tree}{\idxj}}{\tastate}
                \tastateseq_1,
                \tastate_1, \ldots, \tastate_{\idxi-1},
                \tastate',
                \tastate_{\idxi+1}, \ldots, \tastate_\numof,
                \tastateseq_2 \ .
            \end{multline*}
    \end{enumerate}
\end{definition}

Note that item \ref{item:order-n-states} of the definition of $\smodels$
contains a vacuity in that there may be no
$\stack_1, \ldots, \stack_\numof$
satisfying the antecedent (in which case all stacks would be in the meaning of
$\sastate$).  Hence, we require a non-redundancy condition on the automata.

\begin{definition}[Non-Redundancy]
    An order-$\maxord$ annotated stack tree automaton
    \[
        \ta = \tup{
                   \tastates, \sastates_\maxord,\ldots,\sastates_1,
                   \salphabet,
                   \tadelta, \sadelta_\maxord,\ldots,\sadelta_1,
                   \controls,
                   \tafinals, \safinals_\maxord,\ldots,\safinals_1
               }
    \]
    is \emph{non-redundant} if for all
    $\tastate \in \tastates$
    we have that either $\tastate$ has no-incoming transitions, or there exist
    \[
        \tatran{\tastate}{1}{\numof}{\tastate_1}{\sastate_1},
        \ldots,
        \tatran{\tastate}{\numof}{\numof}{\tastate_\numof}{\sastate_\numof}
        \in
        \tadelta
    \]
    such that for all
    $\tree \tmodels{\tenv} \tastateseq_1, \tastate, \tastateseq_2$
    there exist
    $\stack_1, \ldots, \stack_\numof$
    such that
    \[
        \tree
        \treeplus{\idxj}
        \tup{\tastate_1, \stack_1},
        \ldots,
        \tup{\tastate_\numof, \stack_\numof}
        \tmodelsm{\tenv}{\tleaf{\tree}{\idxj}}{\tastate}
        \tastateseq_1,
        \tastate_1, \ldots, \tastate_\numof,
        \tastateseq_2
    \]
    where $\idxj = \seqlen{\tastateseq_1} + 1$.
\end{definition}

This property can be easily satisfied in $\ta_0$ by removing states $\tastate$
that do not satisfy the non-redundancy conditions (this does not change the
language since there were no trees that could be accepted using $\tastate$).
We show later that the property is maintained by saturation.

\subsection{Soundness of a Transition}

After assigning meanings to states, we can define a notion of soundness for the
transitions of the automata.  Intuitively, a transition is sound if it respects
the meanings of its source and target states.

One may derive some more intuition by considering a transition
$q \xrightarrow{a} q'$
of a finite word automaton.  The transition would be sound if, for every word
$w$ in the meaning of $q'$, the same word with an $a$ in front is in the
meaning of $q$.  That is, the transition is sound if an $a$ can appear on
anything accepted from $q'$.  The following definition translates the same idea
to the case of stack trees.

\begin{definition}[Soundness of transitions]
    There are two cases given below.
    \begin{enumerate}
        \item
            A transition
            $\satranfull{\sastate_\midord}
                        {\cha}
                        {\sastateset_\branch}
                        {\sastateset_1, \ldots, \sastateset_\midord}$
            is sound if for any
            $\stack_1 \smodels \sastateset_1$,
            \ldots,
            $\stack_\midord \smodels \sastateset_\midord$
            and
            $\stack_\branch \smodels \sastateset_\branch$
            we have
            $\annot{\cha}{\stack_\branch}
             \scomp{1}
             \stack_1
             \scomp{2}
             \cdots
             \scomp{\midord}
             \stack_\midord
             \smodels
             \sastate_\midord$.

        \item
            A transition
            \[
                \tatranfull{\tastate}
                           {\idxi}
                           {\numof}
                           {\tastate'}
                           {\cha}
                           {\sastateset_\branch}
                           {\sastateset_1, \ldots, \sastateset_\maxord},
            \]
            is sound if for all trees
            $\tree \tmodels{\tenv} \tastateseq_1, \tastate, \tastateseq_2$
            and stacks
            $\stack_1 \smodels \sastateset_1$,
            \ldots
            $\stack_\numof \smodels \sastateset_\numof$,
            and
            $\stack_\branch \smodels \sastateset_\branch$
            and for all
            \[
                \tatran{\tastate}{1}{\numof}{\tastate_1}{\sastate_1},
                \ldots,
                \tatran{\tastate}{\numof}{\numof}{\tastate_\numof}{\sastate_\numof}
            \]
            and stacks
            $\stack'_1, \ldots, \stack'_\numof$
            such that
            \[
                \tree
                \treeplus{\idxj}
                \tup{\tastate_1, \stack'_1},
                \ldots,
                \tup{\tastate_\numof, \stack'_\numof}
                \tmodelsm{\tenv}{\tleaf{\tree}{\idxj}}{\tastate}
                \tastateseq_1,
                \tastate_1, \ldots, \tastate_\numof,
                \tastateseq_2
            \]
            where
            $\idxj = \seqlen{\tastateseq_1} + 1$,
            we have
            \begin{multline*}
                \tree
                \treeplus{\idxj}
                \tup{\tastate_1, \stack'_1},
                \ldots,
                \tup{\tastate_{\idxi-1}, \stack'_{\idxi-1}},
                \tup{\tastate', \stack},
                \tup{\tastate_{\idxi+1}, \stack'_{\idxi+1}},
                \ldots,
                \tup{\tastate_\numof, \stack'_\numof}
                \\
                \tmodelsm{\tenv}{\tleaf{\tree}{\idxj}}{\tastate}
                \tastateseq_1,
                \tastate_1, \ldots, \tastate_{\idxi-1},
                \tastate',
                \tastate_{\idxi+1}, \ldots, \tastate_\numof,
                \tastateseq_2 \ .
            \end{multline*}
            where
            \[
                \stack = \annot{\cha}{\stack_\branch}
                         \scomp{1}
                         \stack_1
                         \scomp{2}
                         \cdots
                         \scomp{\maxord}
                         \stack_\maxord\ .
            \]
    \end{enumerate}
\end{definition}

In the proof, we will have to show that saturation builds a sound automaton.
This means proving soundness for each new transition.  The following lemma shows
that it suffices to only show soundness for the outer collections of
transitions.

\begin{namedlemma}{lem:sound-cascade}{Cascading Soundness}
    If a transition
    \[
        \tatranfull{\tastate}
                   {\idxi}
                   {\numof}
                   {\tastate'}
                   {\cha}
                   {\sastateset_\branch}
                   {\sastateset_1, \ldots, \sastateset_\maxord},
    \]
    is sound then all transitions
    $\satranfull{\sastate_\midord}
                {\cha}
                {\sastateset_\branch}
                {\sastateset_1, \ldots, \sastateset_\midord}$
    appearing within the transition are also sound.
\end{namedlemma}
\begin{proof}
    We march by induction.  Initially
    $\midord = \maxord$
    and we have
    $\satranfull{\sastate}
                {\cha}
                {\sastateset_\branch}
                {\sastateset_1, \ldots, \sastateset_\maxord}$
    where
    $\tatran{\tastate}{\idxi}{\numof}{\tastate_\idxi}{\sastate}$.
    To prove soundness of the transition from $\sastate$, take
    $\stack_1 \smodels \sastateset^\idxi_1$,
    \ldots,
    $\stack_\maxord \smodels \sastateset^\idxi_\maxord$,
    and
    $\stack_\branch \smodels \sastateset^\idxi_\branch$.
    We need to show
    \[
        \stack = \annot{\cha}{\stack_\branch}
                 \scomp{1}
                 \stack_1
                 \scomp{2}
                 \cdots
                 \scomp{\maxord}
                 \stack_\maxord
        \smodels \sastate \ .
    \]
    This is the case if, letting
    $\idxj = \seqlen{\tastateseq_1}$,
    for all transitions
    \[
        \tatran{\tastate}{1}{\numof}{\tastate_1}{\sastate_1},
        \ldots,
        \tatran{\tastate}{\numof}{\numof}{\tastate_\numof}{\sastate_\numof}
    \]
    trees
    $\tree \tmodels{\tenv} \tastateseq_1, \tastate, \tastateseq_2$
    and stacks
    $\stack_1, \ldots, \stack_\numof$
    such that
    \[
        \tree
        \treeplus{\idxj}
        \tup{\tastate_1, \stack_1},
        \ldots,
        \tup{\tastate_\numof, \stack_\numof}
        \tmodelsm{\tenv}{\tleaf{\tree}{\idxj}}{\tastate}
        \tastateseq_1,
        \tastate_1, \ldots, \tastate_\numof,
        \tastateseq_2
    \]
    we have
    \begin{multline*}
        \tree
        \treeplus{\idxj}
        \tup{\tastate_1, \stack_1},
        \ldots,
        \tup{\tastate_{\idxi-1}, \stack_{\idxi-1}},
        \tup{\tastate', \stack},
        \tup{\tastate_{\idxi+1}, \stack_{\idxi+1}},
        \ldots,
        \tup{\tastate_\numof, \stack_\numof}
        \\
        \tmodelsm{\tenv}{\tleaf{\tree}{\idxj}}{\tastate}
        \tastateseq_1,
        \tastate_1, \ldots, \tastate_{\idxi-1},
        \tastate',
        \tastate_{\idxi+1}, \ldots, \tastate_\numof,
        \tastateseq_2 \ .
    \end{multline*}
    These properties are derived immediately from the fact that
    \[
        \tatranfull{\tastate}
                   {\idxi}
                   {\numof}
                   {\tastate'}
                   {\cha}
                   {\sastateset_\branch}
                   {\sastateset_1, \ldots, \sastateset_\maxord},
    \]
    is sound, hence we are done.

    When
    $\midord < \maxord$
    we assume
    $\satranfull{\sastate_{\midord+1}}
                {\cha}
                {\sastateset_\branch}
                {\sastateset_1, \ldots, \sastateset_{\midord+1}}$
    is sound and
    $\sastate_{\midord+1} \satran{\sastate_\midord} \sastateset_{\midord+1}$.
    We show
    $\satranfull{\sastate_\midord}
                {\cha}
                {\sastateset_\branch}
                {\sastateset_1, \ldots, \sastateset_\midord}$
    is also sound.  For this, we take any stacks
    $\stack_1 \smodels \sastateset_1$,
    \ldots
    $\stack_\midord \smodels \sastateset_\midord$,
    and
    $\stack_\branch \smodels \sastateset_\branch$.
    We need to show
    \[
        \stack = \annot{\cha}{\stack_\branch}
                 \scomp{1}
                 \stack_1
                 \scomp{2}
                 \cdots
                 \scomp{\midord}
                 \stack_\midord
        \smodels \sastate_\midord \ .
    \]
    For this, we need for all
    $\stack' \models \sastateset_{\midord+1}$
    that
    $\stack \scomp{(\midord+1)} \stack' \smodels \sastate_{\midord+1}$.
    From the soundness of
    $\satranfull{\sastate_{\midord+1}}
                {\cha}
                {\sastateset_\branch}
                {\sastateset_1, \ldots, \sastateset_{\midord+1}}$
    we have
    \[
        \stack \scomp{(\midord+1)} \stack'
        =
        \annot{\cha}{\stack_\branch}
        \scomp{1}
        \stack_1
        \scomp{2}
        \cdots
        \scomp{(\midord+1)}
        \stack_{\midord+1}
        \smodels
        \sastate_{\midord+1}
    \]
    and we are done.
\end{proof}

\subsection{Soundness of Annotated Stack Tree Automata}

We will prove the saturation constructs a sound automaton.
We first define what it means for an automaton to be sound and prove that a
sound automaton only accepts trees in
$\prestar{\gstrs}{\ta_0}$.

\begin{definition}[Soundness of Annotated Stack Tree Automata]
    An annotated stack tree automaton $\ta$ is sound if
    \begin{enumerate}
        \item
            $\ta$ is obtained from $\ta_0$ by adding new initial states to
            $\sastates_1, \ldots, \sastates_\maxord$ and transitions starting at
            initial states, and

        \item
            in $\ta$, all transitions
            \[
                \tatranfull{\tastate}
                           {\idxi}
                           {\numof}
                           {\tastate'}
                           {\cha}
                           {\sastateset_\branch}
                           {\sastateset_1, \ldots, \sastateset_\maxord}
            \]
            and
            \[
                \satranfull{\sastate_\midord}
                           {\cha}
                           {\sastateset_\branch}
                           {\sastateset_1, \ldots, \sastateset_\midord}
            \]
            are sound, and

        \item
            $\ta$ is non-redundant.
    \end{enumerate}
\end{definition}

We show that a sound annotated stack tree automaton can only accept trees
belonging to
$\prestar{\gstrs}{\ta_0}$.
In fact, we prove a more general result.  In the following lemma, note the
particular case where
$\tree \in \tlang{\tastateseq}{\ta}$
and $\tastateseq$ is a sequence of states in $\controls$ then we have
$\tree \in \prestar{\gstrs}{\ta_0}$.
That is,
$\langof{\ta} \subseteq \prestar{\gstrs}{\ta_0}$.

\begin{namedlemma}{lem:sound-trees}{Sound Acceptance}
    Let $\ta$ be a sound annotated stack automaton.  For all
    $\tree \in \tlang{\tastateseq}{\ta}$
    we have
    $\tree \tmodels{\emptyset} \tastateseq$.
\end{namedlemma}

Before we can prove the result about trees, we first prove a related result
about stacks.  This result and proof is taken almost directly from ICALP
2012~\cite{BCHS12}.

\begin{namedlemma}{lem:sound-stacks}{Sound Acceptance of Stacks}
    Let $\ta$ be a sound annotated stack automaton.  If $\ta$ accepts an
    order-$\midord$ stack $\stack$ from
    $\sastate \in \sastates_\midord$
    then
    $\stack \smodels \sastate$.
\end{namedlemma}
\begin{proof}
    We proceed by induction on the size of the stack (where the size of an
    annotated stack is defined to be the size of a tree representing the stack).

    Let $\stack$ be an order-$\midord$ stack accepted from a state
    $\sastate \in \sastates_\midord$.
    We assume that the property holds for any smaller stack.

    If $\stack$ is empty then $\sastate$ is a final state.  Recall that by
    assumption final states are not initial, hence $\sastate$ is not initial.
    It follows that the empty stack is accepted from $\sastate$ in $\ta_0$ and
    hence $\stack \smodels \sastate$.

    If $\stack$ is a non-empty stack of order-$1$, then
    $\stack = \annot{\cha}{\stack_\branch} \scomp{1} \stack_1$.
    As $\stack$ is accepted from $\sastate$, there exists a transition
    $\satranfull{\sastate}{\cha}{\sastateset_\branch}{\sastateset_1}$
    such that
    $\stack_1$ is accepted from $\sastateset_1$ and $\stack_\branch$ is accepted
    from $\sastateset_\branch$.  By induction we have
    $\stack_1 \smodels \sastateset_1$
    and
    $\stack_\branch \smodels \sastateset_\branch$.
    Since the transition is sound, we have
    $\stack \smodels \sastate$.

    If $\stack$ is a non-empty stack of order-$\midord$, then
    $\stack = \stack_{\midord-1} \scomp{\midord} \stack_\midord$.
    As $\stack$ is accepted from $\sastate$, there exists a transition
    $\sastate \satran{\sastate'} \sastateset$
    such that
    $\stack_\midord$ is accepted from $\sastateset$ and
    $\stack_{\midord-1}$
    is accepted from $\sastate'$.  By induction we have
    $\stack_{\midord-1} \smodels \sastate'$
    and
    $\stack_\midord \smodels \sastateset_\midord$.
    Thus, by the definition of
    $\stack_{\midord-1} \smodels \sastate'$
    we also have
    $\stack = \stack_{\midord-1} \scomp{\midord} \stack_\midord
        \smodels \sastate$.
\end{proof}

We are now ready to prove \reflemma{lem:sound-trees}.
\begin{proof}[Proof of \reflemma{lem:sound-trees}]
    We proceed by induction on the number of nodes in the tree.  In the base
    case, we have
    $\tree \in \tlang{\tastatef}{\ta}$
    for some
    $\tastatef \in \tafinals$
    and $\tree$ is empty.  Thus, we immediately have
    $\tree \tmodels{\tenvempty} \tastatef$.

    Thus, take some non-empty
    $\tree \in \tlang{\tastateseq}{\ta}$.
    Let the sequence
    $\tleaf{\tree}{\idxi}, \ldots, \tleaf{\tree}{\idxi+\numof}$
    be the first complete group of siblings that are all leaf nodes and let
    $\tastateseq
     =
     \tastateseq_1,
     \tastate_1, \ldots, \tastate_\numof,
     \tastateseq_2$
    be the decomposition of $\tastateseq$ such that $\tastateseq_1$ is of length
    $(\idxi-1)$.  That is,
    $\tastate_1, \ldots, \tastate_\numof$
    label the identified leaves of $\tree$.  Furthermore, let
    $\stack_1, \ldots, \stack_\numof$
    be the respective stacks labelling these leaves.  Take the set of
    transitions
    \[
        \tatran{\tastate}{1}{\numof}{\tastate_1}{\sastate_1},
        \ldots
        \tatran{\tastate}{\numof}{\numof}{\tastate_{\numof}}{\sastate_{\numof}}
    \]
    that are used in the accepting run of $\tree$ and the identified leaves.
    Let $\tree'$ be the tree obtained by removing
    $\tleaf{\tree}{\idxi}, \ldots, \tleaf{\tree}{\idxi+\numof'}$.
    We have
    $\tree' \in \tlang{\tastateseq_1, \tastate, \tastateseq_2}{\ta}$
    and by induction
    $\tree' \tmodels{\tenvempty} \tastateseq_1, \tastate, \tastateseq_2$.

    Since $\tastate$ has incoming transitions and $\ta$ is non-redundant, we
    know there exists
    \[
        \tatran{\tastate}{1}{\numof}{\tastate'_1}{\sastate'_1},
        \ldots
        \tatran{\tastate}{\numof}{\numof}{\tastate'_{\numof}}{\sastate'_{\numof}}
    \]
    and
    $\stack'_1, \ldots, \stack'_\numof$
    such that
    \[
        \tree'
        \treeplus{\idxi}
        \tup{\tastate'_1, \stack'_1},
        \ldots,
        \tup{\tastate'_\numof, \stack'_\numof}
        \tmodelsm{\tenvempty}{\tleaf{\tree}{\idxi}}{\tastate}
        \tastateseq_1,
        \tastate'_1, \ldots, \tastate'_\numof,
        \tastateseq_2 \ .
    \]

    Since
    $\stack_1 \smodels \sastate_1$
    we infer from the definition of $\smodels$ at $\sastate_1$ that
    \[
        \tree'
        \treeplus{\idxi}
        \tup{\tastate_1, \stack_1},
        \tup{\tastate'_2, \stack'_2},
        \ldots,
        \tup{\tastate'_\numof, \stack'_\numof}
        \tmodelsm{\tenvempty}{\tleaf{\tree}{\idxi}}{\tastate}
        \tastateseq_1,
        \tastate_1,
        \tastate'_2, \ldots, \tastate'_\numof,
        \tastateseq_2 \ .
    \]
    By repeated applications of the above for each $1 < \idxj \leq \numof$, we
    obtain
     \[
        \tree'
        \treeplus{\idxi}
        \tup{\tastate_1, \stack_1},
        \tup{\tastate_2, \stack_2},
        \ldots,
        \tup{\tastate_\numof, \stack_\numof}
        \tmodelsm{\tenvempty}{\tleaf{\tree}{\idxi}}{\tastate}
        \tastateseq_1,
        \tastate_1, \ldots, \tastate_\numof,
        \tastateseq_2 \ .
    \]
    This implies
    $\tree \tmodels{\tenvempty} \tastateseq$
    since
    $\tmodels{\tenvempty}$
    is less restrictive than
    $\tmodelsm{\tenvempty}{\tleaf{\tree}{\idxi}}{\tastate}$.
\end{proof}

\subsection{Soundness of Saturation}

We first prove that $\ta_0$ is sound, and then that saturation maintains the
property.

\begin{namedlemma}{lem:init-sound}{Soundness of $\ta_0$}
    The initial automaton $\ta_0$ is sound.
\end{namedlemma}
\begin{proof}
    It is trivial that $\ta_0$ is obtained from $\ta_0$, and moreover, we assume
    the non-redundancy condition.  Hence, From \reflemma{lem:sound-cascade} we
    only need to prove soundness of non-initial transitions of the form
    \[
        \satranfull{\sastate_\midord}
                   {\cha}
                   {\sastateset_\branch}
                   {\sastateset_1, \ldots, \sastateset_\maxord}
    \]
    and for transitions in $\tadelta$.

    We first show the case for non-initial
    \[
        \satranfull{\sastate_\midord}
                   {\cha}
                   {\sastateset_\branch}
                   {\sastateset_1, \ldots, \sastateset_\maxord}
    \]
    which is the same as in ICALP 2012.  First note that
    $\sastateset_1, \ldots, \sastateset_\maxord$
    and
    $\sastateset_\branch$
    do not contain initial states.  Then we take
    $\stack_1 \smodels \sastateset_1$,
    \ldots
    $\stack_\midord \smodels \sastateset_\midord$
    and
    $\stack_\branch \smodels \sastateset_\branch$.
    We have to show
    $\annot{\cha}{\stack_\branch}
     \scomp{1}
     \stack_1
     \scomp{2}
     \cdots
     \scomp{\midord}
     \stack_\midord
     \smodels
     \sastate_\midord$.
    In particular, since $\sastate_\midord$ is not initial, we only need to
    construct an accepting run.  Since $\sastateset_\idxi$ and
    $\sastateset_\branch$ are not initial, we have accepting runs from these
    states.  Hence, we build immediately the run beginning with
    $\satranfull{\sastate_\midord}
                {\cha}
                {\sastateset_\branch}
                {\sastateset_1, \ldots, \sastateset_\maxord}$.

    We now prove the case for
    \[
        \tatranfull{\tastate}
                   {\idxi}
                   {\numof}
                   {\tastate'}
                   {\cha}
                   {\sastateset_\branch}
                   {\sastateset_1, \ldots, \sastateset_\maxord} \ .
    \]
    Thus, take any
    $\stack_1 \smodels \sastateset_1$,
    \ldots
    $\stack_\numof \smodels \sastateset_\numof$,
    and
    $\stack_\branch \smodels \sastateset_\branch$
    and any tree
    $\tree \tmodels{\tenv} \tastateseq_1, \tastate, \tastateseq_2$
    and, letting
    $\idxj = \seqlen{\tastateseq_1} + 1$, any
    \[
        \tatran{\tastate}{1}{\numof}{\tastate_1}{\sastate_1},
        \ldots
        \tatran{\tastate}{\numof}{\numof}{\tastate_\numof}{\sastate_\numof}
    \]
    and any
    $\stack'_1, \ldots, \stack'_\numof$
    such that
    $\tree
     \treeplus{\idxj}
     \tup{\tastate_1, \stack'_1}
     \ldots
     \tup{\tastate_\numof, \stack'_\numof}
     \tmodelsm{\tenv}{\tleaf{\tree}{\idxj}}{\tastate}
     \tastateseq_1,
     \tastate_1, \ldots, \tastate_\numof,
     \tastateseq_\numof$.
    Since initial states have no incoming transitions, we know $\tastate$ is not
    a control sate.  We thus have a
    run $\run$ from
    $\tree
     \treeplus{\idxj}
     \tup{\tastate_1, \stack'_1}
     \ldots
     \tup{\tastate_\numof, \stack'_\numof}$
    to some
    $\tree' \in \tweaklang{\ta}$
    such that
    $\tleaf{\tree}{\idxj}$
    does not appear as a leaf of any tree in the run.

    To prove soundness we argue that
    \begin{multline}
        \tree
        \treeplus{\idxj}
        \tup{\tastate_1, \stack'_1},
        \ldots,
        \tup{\tastate_{\idxi-1}, \stack'_{\idxi-1}},
        \tup{\tastate', \stack},
        \tup{\tastate_{\idxi+1}, \stack'_{\idxi+1}},
        \ldots,
        \tup{\tastate_\numof, \stack'_\numof}
        \\
        \tmodelsm{\tenv}{\tleaf{\tree}{\idxj}}{\tastate}
        \tastateseq_1,
        \tastate_1, \ldots, \tastate_{\idxi-1},
        \tastate',
        \tastate_{\idxi+1}, \ldots, \tastate_\numof,
        \tastateseq_2
        \label{eqn:initial-soundness-long}
    \end{multline}
    where
    $\stack
     =
     \annot{\cha}{\stack_\branch}
     \scomp{1}
     \stack_1
     \scomp{2}
     \cdots
     \scomp{\maxord}
     \stack_\maxord$.
    To do so, we take the run $\run$ obtained above and build a run $\run'$ by
    removing all operations applied to nodes that are descendants of
    $\tleaf{\tree}{\idxj}\idxi$.
    Observe that $\run'$ can be applied to
    \[
        \tree
        \treeplus{\idxj}
        \tup{\tastate_1, \stack'_1},
        \ldots,
        \tup{\tastate_{\idxi-1}, \stack'_{\idxi-1}},
        \tup{\tastate', \stack},
        \tup{\tastate_{\idxi+1}, \stack'_{\idxi+1}},
        \ldots,
        \tup{\tastate_\numof, \stack'_\numof}
    \]
    since none of the operations apply to a descendant of
    $\tleaf{\tree}{\idxj}\idxi$.
    By applying this run we obtain a tree $\tree''$ which is $\tree'$
    less all nodes that are strict descendants of
    $\tleaf{\tree}{\idxj}\idxi$
    and where
    $\tleaf{\tree}{\idxj}\idxi$
    is labelled by
    $\tup{\tastate', \stack}$.
    Thus, we take the accepting run of $\tree'$ witnessing
    $\tree' \in \tweaklang{\ta_0}$,
    remove all nodes that are strict descendants of
    $\tleaf{\tree}{\idxj}\idxi$
    and label
    $\tleaf{\tree}{\idxj}\idxi$
    by $\tastate'$.  This gives us a run witnessing
    $\tree'' \in \tweaklang{\ta_0}$
    by using
    \[
        \tatranfull{\tastate}
                   {\idxi}
                   {\numof}
                   {\tastate'}
                   {\cha}
                   {\sastateset_\branch}
                   {\sastateset_1, \ldots, \sastateset_\maxord} \ .
    \]
    at
    $\tleaf{\tree}{\idxj}\idxi$
    and the accepting runs from the non-initial $\sastateset_\branch$,
    $\sastateset_1, \ldots, \sastateset_\maxord$.  This gives us
    (\ref{eqn:initial-soundness-long}) as required.
\end{proof}

We now show that, at every stage of saturation, we maintain a sound automaton.

\begin{namedlemma}
      {lem:saturation-soundness-step}
      {Soundness of the Saturation Step}
    Given a sound automaton $\ta$, we have $\ta' = \ap{\satfn}{\ta}$ is sound.
\end{namedlemma}
\begin{proof}
    We analyse all new transitions
    \[
        \tatranfull{\tastate}
                   {\idxi}
                   {\numof}
                   {\control}
                   {\cha}
                   {\sastateset^\mnew_\branch}
                   {\sastateset^\mnew_1, \ldots, \sastateset^\mnew_\maxord} \ .
    \]
    Proving these transitions are sound and do not cause redundancy is sufficient
    via \reflemma{lem:sound-cascade}.

    Let us begin with the transitions introduced by rules that do not remove
    nodes from the tree.  We argue that for all trees
    $\tree \tmodels{\tenv} \tastateseq_1, \tastate, \tastateseq_2$
    and stacks
    $\stack_1 \smodels \sastateset^\mnew_1$,
    \ldots
    $\stack_\numof \smodels \sastateset^\mnew_\numof$,
    and
    $\stack_\branch \smodels \sastateset^\mnew_\branch$
    and for all
    \[
        \tatran{\tastate}{1}{\numof}{\tastate_1}{\sastate_1},
        \ldots,
        \tatran{\tastate}{\numof}{\numof}{\tastate_\numof}{\sastate_\numof}
    \]
    and stacks
    $\stack'_1, \ldots, \stack'_\numof$
    such that
    \[
        \tree
        \treeplus{\idxj}
        \tup{\tastate_1, \stack'_1},
        \ldots,
        \tup{\tastate_\numof, \stack'_\numof}
        \tmodelsm{\tenv}{\tleaf{\tree}{\idxj}}{\tastate}
        \tastateseq_1,
        \tastate_1, \ldots, \tastate_\numof,
        \tastateseq_2
    \]
    where
    $\idxj = \seqlen{\tastateseq_1} + 1$
    we have, letting
    \[
        \tree_1 =
        \tree
        \treeplus{\idxj}
        \tup{\tastate_1, \stack'_1},
        \ldots,
        \tup{\tastate_{\idxi-1}, \stack'_{\idxi-1}},
        \tup{\control, \stack},
        \tup{\tastate_{\idxi+1}, \stack'_{\idxi+1}},
        \ldots,
        \tup{\tastate_\numof, \stack'_\numof}
    \]
    and
    $\tastateseq'_1 = \tastateseq_1, \tastate_1, \ldots, \tastate_{\idxi-1}$
    and
    $\tastateseq'_2 = \tastate_{\idxi+1}, \ldots, \tastate_\numof, \tastateseq_2$
    that
    \begin{equation}
        \label{eqn:soundness-step-prop}
        \tree_1
        \tmodelsm{\tenv}{\tleaf{\tree}{\idxj}}{\tastate}
        \tastateseq'_1,
        \control,
        \tastateseq'_2
    \end{equation}
    where
    $\stack = \annot{\cha}{\stack_\branch}
              \scomp{1}
              \stack_1
              \scomp{2}
              \cdots
              \scomp{\maxord}
              \stack_\maxord$.

    We proceed by a case distinction on the rule $\strule$ which led to the
    introduction of the new transition.  In each case, let
    $\tree_2 \in \ap{\strule}{\tree_1}$
    be the result of applying $\strule$ at node
    $\tleaf{\tree}{\idxj}\idxi$.
    In all cases except when $\strule$ removes nodes, $\tastate$ already has an
    incoming transition, hence we do not need to argue non-redundancy (since
    $\ta$ is non-redundant).
    \begin{itemize}
        \item
            When
            $\strule = \gtrule{\control'}{\srew{\chb}{\cha}}{\control}$
            we derived the new transition from some transtion
            \[
                \tatranfull{\tastate}
                           {\idxi}
                           {\numof}
                           {\control'}
                           {\chb}
                           {\sastateset^\mnew_\branch}
                           {\sastateset^\mnew_1, \ldots, \sastateset^\mnew_\maxord}
            \]
            and since this transition is sound
            $\tree_2
             \tmodelsm{\tenv}{\tleaf{\tree}{\idxj}}{\tastate}
             \tastateseq'_1, \control, \tastateseq'_2$.
            We take the run witnessing soundness for $\tree_2$ and prepend the
            application of $\strule$ to $\tree_1$.  This gives us a run
            witnessing (\ref{eqn:soundness-step-prop}) as required.

        \item
            When
            $\strule = \gtrule{\control'}{\scpush{\midord}}{\control}$, then
            when $\midord > 1$ we derived the new transition from some
            \[
                \tatranfull{\tastate}
                           {\idxi}
                           {\numof}
                           {\control'}
                           {\cha}
                           {\sastateset_\branch}
                           {\sastateset_1,
                            \sastateset_2, \ldots, \sastateset_\maxord}
            \]
            and
            $\sastateset_1 \satrancol{\cha}{\sastateset'_\branch} \sastateset'_1$
            and the new transition is of the form
            \[
                \tatranfull{\tastate}
                           {\idxj}
                           {\numof}
                           {\control}
                           {\cha}
                           {\sastateset'_\branch}
                           {\sastateset'_1,
                            \sastateset_2,
                            \ldots,
                            \sastateset_{\midord-1},
                            \sastateset_\midord \cup \sastateset_\branch,
                            \sastateset_{\midord+1},
                            \ldots,
                            \sastateset_\maxord}
            \]
            Furthermore, we have $\tree_2$ has at
            $\tleaf{\tree}{\idxj}\idxi$ the stack
            \[
                \annot{\cha}{\stack_\midord}
                \scomp{1}
                \annot{\cha}{\stack_\branch}
                \scomp{1}
                \stack_1
                \scomp{2}
                \cdots
                \scomp{\maxord}
                \stack_\maxord
            \]
            and we have
            $\stack_\midord
             \smodels
             \sastateset^\mnew_\midord
             =
             \sastateset_\midord \cup \sastateset_\branch$
            and
            $\stack_1 \smodels \sastateset^\mnew_1 = \sastateset'_1$
            and from soundness of
            $\sastateset_1 \satrancol{\cha}{\sastateset'_\branch} \sastateset'_1$
            we have
            $\annot{\cha}{\stack_\branch}
                   \scomp{1}
                   \stack_1
             \smodels \sastateset_1$.
            Thus, we can apply soundness of the transition from $\control'$ to
            obtain
            $\tree_2
             \tmodelsm{\tenv}{\tleaf{\tree}{\idxj}}{\tastate}
             \tastateseq'_1, \control', \tastateseq'_2$.
            We prepend to the run witnessing this property an application of
            $\strule$ to $\tree_1$ at node
            $\tleaf{\tree}{\idxj}\idxi$
            to obtain a run witnessing (\ref{eqn:soundness-step-prop}) as
            required.

            When $\midord = 1$ we began with a transition
            \[
                \tatranfull{\tastate}
                           {\idxi}
                           {\numof}
                           {\control'}
                           {\cha}
                           {\sastateset_\branch}
                           {\sastateset_1,
                            \sastateset_2 \ldots, \sastateset_\maxord}
            \]
            and
            $\sastateset_1 \satrancol{\cha}{\sastateset'_\branch} \sastateset'_1$
            and the new transition is of the form
            \[
                \tatranfull{\tastate}
                            {\idxj}
                            {\numof}
                            {\control}
                            {\cha}
                            {\sastateset'_\branch}
                            {\sastateset'_1 \cup \sastateset_\branch,
                             \sastateset_2,
                             \ldots,
                             \sastateset_\maxord} \ .
            \]
            Furthermore, we have $\tree_2$ has at
            $\tleaf{\tree}{\idxj}\idxi$ the stack
            \[
                \annot{\cha}{\stack_1}
                \scomp{1}
                \annot{\cha}{\stack_\branch}
                \scomp{1}
                \stack_1
                \scomp{2}
                \cdots
                \scomp{\maxord}
                \stack_\maxord
            \]
            and we have
            $\stack_1
             \smodels
             \sastateset^\mnew_1
                = \sastateset'_1 \cup \sastateset_\branch$
            and from
            $\stack_\branch
             \smodels
             \sastateset^\mnew_\branch = \sastateset'_\branch$
            and soundness of
            $\sastateset_1 \satrancol{\cha}{\sastateset'_\branch} \sastateset'_1$
            we have
            $\annot{\cha}{\stack_\branch}
                   \scomp{1}
                   \stack_1
             \smodels \sastateset'_1$.
            Thus, we can apply soundness of the transition from $\control'$
            using
            $\stack_1 \smodels \sastateset_\branch$
            (since
            $\stack_1 \smodels \sastateset^\mnew_1
                = \sastateset'_1 \cup \sastateset_\branch$)
            to obtain
            $\tree_2
             \tmodelsm{\tenv}{\tleaf{\tree}{\idxj}}{\tastate}
             \tastateseq'_1, \control', \tastateseq'_2$.
            We prepend to the run witnessing this property an application of
            $\strule$ to $\tree_1$ at node
            $\tleaf{\tree}{\idxj}\idxi$
            to obtain a run witnessing (\ref{eqn:soundness-step-prop}) as
            required.

        \item
            When
            $\strule = \gtrule{\control}{\spush{\midord}}{\control'}$
            we started with a transition
            \[
                \tatranfull{\tastate}
                           {\idxi}
                           {\numof}
                           {\control'}
                           {\cha}
                           {\sastateset_\branch}
                           {\sastateset_1, \ldots, \sastateset_\maxord}
            \]
            and
            $\satranfull{\sastateset_\midord}
                        {\cha}
                        {\sastateset'_\branch}
                        {\sastateset'_1, \ldots, \sastateset'_\midord}$
            and the new transition is of the form
            \[
                \tatranfull{\tastate}
                           {\idxj}
                           {\numof}
                           {\control}
                           {\cha}
                           {\sastateset_\branch \cup \sastateset_\branch}
                           {\sastateset_1 \cup \sastateset'_1,
                            \ldots,
                            \sastateset_{\midord-1} \cup \sastateset'_{\midord-1},
                            \sastateset'_\midord,
                            \sastateset_{\midord+1},
                            \ldots,
                            \sastateset_\maxord} \ .
            \]
            Let
            $\stack' = \annot{\cha}{\stack_\branch}
                       \scomp{1}
                       \stack_1
                       \scomp{2}
                       \cdots
                       \scomp{\midord-1}
                       \stack_{\midord-1}$,
            we have that $\tree_2$ has at node
            $\tleaf{\tree}{\idxj}\idxi$ the stack
            \[
                \annot{\cha}{\stack_\branch}
                \scomp{1}
                \stack_1
                \scomp{2}
                \cdots
                \scomp{(\midord-1)}
                \stack_{\midord-1}
                \scomp{\midord}
                \stack'
                \scomp{\midord}
                \stack_{\midord+1}
                \scomp{(\midord+1)}
                \cdots
                \scomp{\maxord}
                \stack_\maxord \ .
            \]
            Note, by assumption we have
            $\stack_1 \smodels \sastateset^\mnew_1
                = \sastateset_1 \cup \sastateset'_1$,
            \ldots,
            $\stack_{\midord-1}
             \smodels
             \sastateset^\mnew_{\midord-1}
             =
             \sastateset_{\midord-1} \cup \sastateset'_{\midord-1}$
            and
            $\stack_\branch
             \smodels
             \sastateset^\mnew_\branch
             =
             \sastateset_\branch \cup \sastateset'_\branch$.
            Thus from soundness of
            $\satranfull{\sastateset_\midord}
                        {\cha}
                        {\sastateset'_\branch}
                        {\sastateset'_1, \ldots, \sastateset'_\midord}$
            we have
            $\stack' \smodels \sastateset_\midord$.
            Consequently, from the soundness of the transition from $\control'$
            we have
            $\tree_2
             \tmodelsm{\tenv}{\tleaf{\tree}{\idxj}}{\tastate}
             \tastateseq'_1, \control', \tastateseq'_2$.
            We prepend to the run witnessing this property an application of
            $\strule$ to $\tree_1$ at node
            $\tleaf{\tree}{\idxj}\idxi$
            to obtain a run witnessing (\ref{eqn:soundness-step-prop}) as
            required.

        \item
            When
            $\strule = \gtrule{\control}{\spop{\midord}}{\control'}$
            we derived the new transition from
            \[
                \tatranfullk{\tastate}
                            {\idxi}
                            {\numof}
                            {\control'}
                            {\sastate_\midord}
                            {\sastateset_{\midord+1},
                             \ldots,
                             \sastateset_\maxord}
            \]
            and the new transition is of the form
            \[
                \tatranfull{\tastate}
                           {\idxj}
                           {\numof}
                           {\control}
                           {\cha}
                           {\emptyset}
                           {\emptyset,
                            \ldots,
                            \emptyset,
                            \set{\sastate_\midord},
                            \sastateset_{\midord+1},
                            \ldots,
                            \sastateset_\maxord}
            \]
            The tree $\tree_2$ has labelling
            $\tleaf{\tree}{\idxj}\idxi$
            the stack
            $\stack' = \stack_\midord
                       \scomp{(\midord+1)}
                       \cdots
                       \scomp{\maxord}
                       \stack_\maxord$
            and since
            $\stack_{\midord+1} \smodels \sastateset_{\midord+1}$,
            \ldots,
            $\stack_{\maxord} \smodels \sastateset_{\maxord}$
            we have from the definition of $\tmodels{\tenv}$ and
            $\stack_\midord \smodels \sastate_\midord$
            that
            $\tree_2
             \tmodelsm{\tenv}{\tleaf{\tree}{\idxj}}{\tastate}
             \tastateseq'_1, \control', \tastateseq'_2$.
            As before, we prepend to the run witnessing this property an
            application of $\strule$ to $\tree_1$ at node
            $\tleaf{\tree}{\idxj}\idxi$
            to obtain a run witnessing (\ref{eqn:soundness-step-prop}) as
            required.

        \item
            When
            $\strule = \gtrule{\control}{\scollapse{\midord}}{\control'}$
            we began with a transition
            \[
                \tatranfullk{\tastate}
                            {\idxi}
                            {\numof}
                            {\control'}
                            {\sastate_\midord}
                            {\sastateset_{\midord+1}, \ldots, \sastateset_\maxord}
            \]
            and the new transition has the form
            \[
                \tatranfull{\tastate}
                           {\idxj}
                           {\numof}
                           {\control}
                           {\cha}
                           {\set{\sastate_\midord}}
                           {\emptyset,
                            \ldots,
                            \emptyset,
                            \sastateset_{\midord+1},
                            \ldots,
                            \sastateset_\maxord}
            \]
            The tree $\tree_2$ has labelling
            $\tleaf{\tree}{\idxj}\idxi$
            the stack
            $\stack' = \stack_\branch
                       \scomp{(\midord+1)}
                       \stack_{\midord+1}
                       \scomp{(\midord+2)}
                       \cdots
                       \scomp{\maxord}
                       \stack_\maxord$
            and since
            $\stack_{\midord+1} \smodels \sastateset_{\midord+1}$,
            \ldots,
            $\stack_{\maxord} \smodels \sastateset_{\maxord}$
            we have from the definition of $\tmodels{\tenv}$ and
            $\stack_\branch \smodels \sastate_\midord$
            that
            $\tree_2
             \tmodelsm{\tenv}{\tleaf{\tree}{\idxj}}{\tastate}
             \tastateseq'_1, \control', \tastateseq'_2$.
            As before, we prepend to the run witnessing this property an
            application of $\strule$ to $\tree_1$ at node
            $\tleaf{\tree}{\idxj}\idxi$
            to obtain a run witnessing (\ref{eqn:soundness-step-prop}) as
            required.

        \item
            When
            $\strule = \stpush{\control}{\control_1, \ldots, \control_{\numof'}}$
            we had transitions
            \[
                \tatranfull{\tastate}
                           {\idxi}
                           {\numof}
                           {\tastate'}
                           {\cha}
                           {\sastateset_\branch}
                           {\sastateset_1, \ldots, \sastateset_\maxord}
            \]
            and
            \[
                \tatranfull{\tastate'}
                           {1}
                           {\numof'}
                           {\control_1}
                           {\cha}
                           {\sastateset^1_\branch}
                           {\sastateset^1_1, \ldots, \sastateset^1_\maxord},
                \ldots,
                \tatranfull{\tastate'}
                           {\numof'}
                           {\numof'}
                           {\control_{\numof'}}
                           {\cha}
                           {\sastateset^{\numof'}_\branch}
                           {\sastateset^{\numof'}_1,
                            \ldots,
                            \sastateset^{\numof'}_\maxord}
            \]
            and the new transition added is of the form
            \[
                \tatranfull{\tastate}
                           {\idxi}
                           {\numof}
                           {\control}
                           {\cha}
                           {\sastateset^\text{new}_\branch}
                           {\sastateset^\text{new}_1, \ldots, \sastateset^\text{new}_\maxord}
            \]
            where
            $\sastateset^\text{new}_\branch =
             \sastateset_\branch
             \cup
             \sastateset^1_\branch
             \cup
             \cdots
             \cup
             \sastateset^{\numof'}_\branch$
            and for all $\midord$, we have
            $\sastateset^\text{new}_\midord =
             \sastateset_1
             \cup
             \sastateset^1_\midord
             \cup
             \cdots
             \cup
             \sastateset^{\numof'}_\midord$.
            Letting $\tree'_1 =$
            \[
                \tree
                \treeplus{\idxj}
                \tup{\tastate_1, \stack'_1},
                \ldots,
                \tup{\tastate_{\idxi-1}, \stack'_{\idxi-1}},
                \tup{\tastate', \stack},
                \tup{\tastate_{\idxi+1}, \stack'_{\idxi+1}},
                \ldots,
                \tup{\tastate_{\numof}, \stack'_{\numof}}
            \]
            and
            $\tenv' = \envmod{\tenv}{\tleaf{\tree}{\idxj}}{\tastate}$
            we have from
            $\sastateset^\mnew_\branch
                = \sastateset_\branch \cup
                  \sastateset^1_\branch
                  \cup \cdots \cup
                  \sastateset^{\numof'}_\branch$
            and
            $\sastateset^\mnew_1
                = \sastateset_1 \cup
                  \sastateset^1_1
                  \cup \cdots \cup
                  \sastateset^{\numof'}_1$,
            \ldots,
            $\sastateset^\mnew_\maxord
                = \sastateset_\maxord \cup
                  \sastateset^1_\maxord
                  \cup \cdots \cup
                  \sastateset^{\numof'}_\maxord$,
            and by soundness of the transition from $\tastate'$ that
            $\tree'_1
             \tmodels{\tenv'}
             \tastateseq'_1, \tastate', \tastateseq'_2$.
            Thus, from non-redundancy
            and repeated applications of the soundness of the
            transition from $\control_1$ to the
            soundness from $\control_{\numof'}$ (as in the proof of
            \reflemma{lem:sound-trees}) we have
            \[
                \tree_2
                =
                \tree'_1
                \treeplus{(\idxj+\idxi)}
                \tup{\control_1, \stack},
                \ldots,
                \tup{\control_{\numof'}, \stack}
                \tmodelsm{\tenv'}{\tleaf{\tree}{\idxj}\idxi}{\tastate'}
                \tastateseq'_1, \control_1, \ldots, \control_{\numof'}, \tastateseq'_2 \ .
            \]
            We prepend to the run witnessing this property an
            application of $\strule$ to $\tree_1$ at node
            $\tleaf{\tree}{\idxj}\idxi$
            to obtain a run witnessing (\ref{eqn:soundness-step-prop}) as
            required.
    \end{itemize}
    The remaining case is for the operations that remove nodes from the tree.
    For $\stpop{\control_1, \ldots, \control_\numof}{\control}$ we introduced
    \[
        \tatranfull{\control}
                   {1}
                   {\numof}
                   {\control_1}
                   {\cha}
                   {\emptyset}
                   {\emptyset, \ldots, \emptyset}
    \]
    to
    \[
        \tatranfull{\control}
                   {\numof}
                   {\numof}
                   {\control_\numof}
                   {\cha}
                   {\emptyset}
                   {\emptyset, \ldots, \emptyset} \ .
    \]
    We prove soundness of the first of these rules, with the others being
    symmetrical.  Taking any sequence of transitions
    \[
        \tatran{\control}{1}{\numof}{\tastate_1}{\sastate_1},
        \ldots,
        \tatran{\control}{\numof}{\numof}{\tastate_\numof}{\sastate_\numof}
    \]
    any
    $\tree \tmodels{\tenv} \tastateseq_1, \control, \tastateseq_2$
    and $\stack_1$, \ldots, $\stack_\numof$ such that, letting
    $\idxj = \seqlen{\tastateseq_1} + 1$,
    \[
        \tree' =
        \tree
        \treeplus{\idxj}
        \tup{\tastate_1, \stack_1},
        \ldots,
        \tup{\tastate_\numof, \stack_\numof}
        \tmodelsm{\tenv}{\tleaf{\tree}{\idxj}}{\control}
        \tastateseq_1,
        \tastate_1, \ldots, \tastate_\numof,
        \tastateseq_2 \ .
    \]
    We need to show for any stack with top character $\cha$ that
    \[
        \tree
        \treeplus{\idxj}
        \tup{\control_1, \stack},
        \tup{\tastate_2, \stack_2},
        \ldots,
        \tup{\tastate_\numof, \stack_\numof}
        \tmodelsm{\tenv}{\tleaf{\tree}{\idxj}}{\control}
        \tastateseq_1,
        \control_1, \tastate_2, \ldots, \tastate_\numof,
        \tastateseq_2 \ .
    \]

    Take the run witnessing the property for $\tree'$.  This must
    necessarily pass some tree where
    $\tleaf{\tree}{\idxj}$
    is exposed and contains control state $\control$.  Moreover, this is
    the first such exposure of the node.  Since we assume, for all
    $\control$, there is only one rule
    $\stpop{\control'_1, \ldots, \control'_2}{\control}$
    for any
    $\control'_1, \ldots, \control'_\numof$,
    the node must be exposed by an application of $\strule$.

    Thus, we can remove from the run all operations applied to a descendant of
    $\tleaf{\tree}{\idxj}1$ before its exposure.  This run then can be applied
    to
    \[
        \tree
        \treeplus{\idxj}
        \tup{\control_1, \stack},
        \tup{\tastate_2, \stack_2}
        \ldots,
        \tup{\tastate_\numof, \stack_\numof}
    \]
    to witness
    $\tree
     \treeplus{\idxj}
     \tup{\control_1, \stack},
     \tup{\tastate_2, \stack_2},
     \ldots,
     \tup{\tastate_\numof, \stack_\numof}
     \tmodelsm{\tenv}{\tleaf{\tree}{\idxj}}{\control}
     \tastateseq_1,
     \control_1, \tastate_2, \ldots, \tastate_\numof,
     \tastateseq_2$.

    To prove non-redundancy, we simply take any stacks $\stack_1$, \ldots,
    $\stack_\numof$ and apply $\strule$ to
    $\tree
     \treeplus{\idxj}
     \tup{\control_1, \stack_1},
     \ldots,
     \tup{\control_\numof, \stack_\numof}$
    to obtain $\tree$ from which the remainder of the run exists by
    assumption.
\end{proof}

\begin{namedlemma}{lem:soundness}{Soundness of Saturation}
    The automaton $\ta$ obtained by saturation from $\ta_0$ is such that
    $\langof{\ta} \subseteq \prestar{\gstrs}{\ta_0}$.
\end{namedlemma}
\begin{proof}
    By \reflemma{lem:init-sound} we have that $\ta_0$ is sound.  Thus, by
    induction, assume $\ta$ is sound.  We have
    $\ta' = \ap{\satfn}{\ta}$
    and by \reflemma{lem:saturation-soundness-step} we have that $\ta'$ is
    sound.

    Thus, the $\ta$ that is the fixed point of saturation is sound, and we have
    from \reflemma{lem:sound-trees} that $\langof{\ta} \subseteq
    \prestar{\gstrs}{\ta_0}$.
\end{proof}

\section{Lower Bounds on the Reachability Problem}
\label{sec:lower-bound}

We show that that global backwards reachability problem is
$\maxord$-EXPTIME-hard for an order-$\maxord$ GASTRS.  The proof is by reduction
from the $\maxord$-EXPTIME-hardness of determining the winner in an
order-$\maxord$ reachability game~\cite{CW07}.

\begin{proposition}[Lower Bound]
    The global backwards reachability problem for order-$\maxord$ GASTRSs is
    $\maxord$-EXPTIME-hard.
\end{proposition}
\begin{proof}
    We reduce from the problem of determining the winner in an order-$\maxord$
    pushdown reachability game~\cite{CW07}.

    We first need to define higher-order stacks and their operations.
    Essentially, they are just annotated stacks without collapse.  That is
    order-$1$ stacks are of the form
    $\kstack{1}{\cha_1 \ldots \cha_\numof}$
    where
    $\cha_1 \ldots \cha_\numof \in \salphabet^\ast$.  Order-$\midord$ stacks for
    $\midord > 1$ are of the form
    $\kstack{\midord}{\stack_1 \ldots \stack_\numof}$
    where
    $\stack_1, \ldots, \stack_\numof$
    are order-$(\midord-1)$ stacks.

    Their operations are
    \[
        \hostackops{\maxord}{\cha}
        =
        \setcomp{\hocpush{\cha}}{\cha \in \salphabet}
        \cup
        \setcomp{\spush{\midord}}{2 \leq \midord \leq \maxord}
        \cup
        \setcomp{\spop{\midord}}{1 \leq \midord \leq \maxord} \ .
    \]
    The
    $\spush{\midord}$ and
    $\spop{\midord}$
    operations are analogous to annotated stacks.  We define
    $\ap{\hocpush{\cha}}{\stack} = \cha \scomp{1} \stack$.

    Such a game is defined as a tuple
    $\tup{\controls, \salphabet, \rules, \finals}$
    where
    $\controls = \controls_1 \cup \controls_2$
    is a finite set of control states partitioned into those belonging to player
    1 and player 2 respectively, $\salphabet$ is a finite set of stack
    characters,
    $\rules \subseteq \controls \times
                      \salphabet \times
                      \hostackops{\maxord}{\salphabet} \times
                      \controls$
    is a finite set of transition rules, and
    $\finals \subseteq \controls$
    is a set of target control states.

    Without loss of generality, we assume that for all
    $\control \in \controls_2$
    and
    $\cha \in \salphabet$
    there exactly two rules in $\rules$ of the form
    $\horule{\control}{\cha}{\sop}{\control'}$
    for some $\sop$ and $\control'$.

    A configuration is a tuple
    $\config{\control}{\stack}$
    of a control state and higher-order stack.
    A winning play of a game from an initial configuration
    $\config{\control_0}{\stack_0}$
    for player 1 is a tree labelled by configurations such that
    \begin{itemize}
        \item
            all leaf nodes are labelled by configurations
            $\config{\control}{\stack}$
            with
            $\control \in \finals$.
        \item
            if an internal node is labelled
            $\config{\control}{\stack}$
            with
            $\control \in \controls_1$
            then the node has one child labelled by
            $\config{\control'}{\stack'}$
            such that for some
            $\horule{\control}{\cha}{\sop}{\control'} \in \rules$
            we have
            $\stack = \cha \scomp{1} \stack''$
            for some $\stack''$ and
            $\stack' = \ap{\sop}{\stack}$.
        \item
            if an internal node is labelled
            $\config{\control}{\stack}$
            with
            $\control \in \controls_2$
            then when
            $\stack = \cha \scomp{1} \stack'$
            for some $\stack'$ and we have the rules
            $\horule{\control}{\cha}{\sop_1}{\control_1}$,
            and
            $\horule{\control}{\cha}{\sop_2}{\control_2}$,
            then the node has two children labelled by
            $\config{\control_1}{\stack_1}$
            and
            $\config{\control_2}{\stack_2}$
            with
            $\stack_1 = \ap{\sop_1}{\stack}$
            and
            $\stack_1 = \ap{\sop_1}{\stack}$.
    \end{itemize}
    Note, we assume that the players can always apply all available rules for a
    given $\control$ and $\cha$ in the game (unless a control in $\finals$ is
    reached).  This is standard and can be done with the use of a
    ``bottom-of-stack'' marker at each order.

    Determining if player 1 wins the game is known to be $\maxord$-EXPTIME
    hard~\cite{CW07}.  This amounts to asking whether a winning game tree can be
    constructed from the initial configuration
    $\config{\control_0}{\stack_0}$.

    That the winning game trees are regular can be easily seen: we simply assert
    that all leaf nodes are labelled by some
    $\control \in \finals$.

    We build a GASTRS that constructs play trees.  We simulate a move in the
    game via several steps in the GASTRS, hence its control states will contain
    several copies of the control states of the game.  Suppose we have a rule
    $\horule{\control}{\cha}{\sop}{\control'}$
    where $\control \in \controls_1$.
    The first step in the simulation will be to check that the top character is
    $\cha$, for which we will use
    $\gtrule{\control}{\srew{\cha}{\cha}}{\ccopy{\control}{1}}$
    where
    $\ccopy{\control}{1}$
    is a new control state.  The next step will create a new node in the play
    tree using
    $\stpush{\ccopy{\control}{1}}{\ccopy{\control'}{2}}$
    which uses the intermediate control state
    $\ccopy{\control'}{2}$.
    The final step is to apply the stack operation and move to $\control'$.
    When
    $\sop = \spush{\midord}$
    or
    $\sop = \spop{\midord}$
    we can use
    $\gtrule{\ccopy{\control'}{2}}{\sop}{\control'}$.
    When
    $\sop = \hocpush{\chb}$
    we use another intermediate control state and
    $\gtrule{\ccopy{\control'}{2}}{\scpush{1}}{\ccopy{\control'}{3}}$
    and
    $\gtrule{\ccopy{\control'}{3}}{\srew{\cha}{\chb}}{\control'}$.

    When
    $\control \in \controls_2$
    with the rules
    $\horule{\control}{\cha}{\sop_1}{\control_1}$
    and
    $\horule{\control}{\cha}{\sop_2}{\control_2}$
    we use
    $\gtrule{\control}{\srew{\cha}{\cha}}{\ccopy{\control}{1}}$,
    \[
        \stpush{\ccopy{\control}{1}}{\ccopy{\control_1}{2}, \ccopy{\control_2}{2}} \ ,
    \]
    and similar rules to the previous case to apply $\sop$ and move to
    $\control_1$ or $\control_2$.

    Let the above GASTRS be $\gstrs$.  From the initial single-node tree
    $\tree_0$ whose node is labelled
    $\tup{\control_0, \stack_0}$
    it is clear that a tree whose leaf nodes are only labelled by control states
    in $\finals$ can be reached iff there is a winning play of player 1 in the
    reachability game.  We can easily build a tree automaton $\ta_0$ that
    accepts only these target trees.  Since checking membership
    $\tree_0 \in \prestar{\gstrs}{\ta_0}$
    is linear in the size of tree automaton representing
    $\prestar{\gstrs}{\ta_0}$
    we obtain our lower bound as required.
\end{proof}

\end{document}